\documentclass[runningheads]{llncs}
\usepackage[T1]{fontenc}
\usepackage{microtype}
\usepackage[hidelinks]{hyperref}
\usepackage{amsmath,amsfonts,amssymb,enumerate}
\let\proof\relax
\let\endproof\relax
\usepackage{amsthm}
\usepackage{xcolor}
\usepackage{cite}
\usepackage{cleveref}
\usepackage{optidef}
\usepackage{algorithm}
\usepackage{algorithmic}
\usepackage{graphicx}

\spnewtheorem{claim}{Claim}{\itshape}{\rmfamily}

\newenvironment{innerproof}
{\proof}
{\endproof}
\newcommand{\pstate}[2]{\langle #1,#2\rangle}

\title{Flow Games with Public Arcs:\\The Least Core and the Nucleolus\thanks{
Supported in part by the National Natural Science Foundation of China (Nos.\,12671377 and 12171444) and Shandong Provincial Natural Science Foundation (No.\, ZR2025MS104).
}}
\titlerunning{Flow Games with Public Arcs}
\author{Tianhang Lu \and Han Xiao \and Qizhi Fang}
\authorrunning{T. Lu et al.}
\institute{School of Mathematical Sciences, Ocean University of China, Qingdao, China\\
\email{tlu@stu.ouc.edu.cn, hxiao@ouc.edu.cn, qfang@ouc.edu.cn}}

\hypersetup{
  pdftitle={Flow Games with Public Arcs: The Least Core and the Nucleolus},
  pdfauthor={Tianhang Lu, Han Xiao, Qizhi Fang}
}

\begin{document}

\maketitle
\begin{abstract}
We study flow games with public arcs, an extension of classical cooperative flow games that allows players to use public resources.
In these games, a coalition corresponds to a set of arcs, while certain arcs, called public arcs, can be used freely by any coalition.
The value of a coalition is the maximum flow value achievable using the arcs controlled by the coalition along with the public arcs.
We investigate two solution concepts, the least core and the nucleolus.
Both solution concepts provide fair ways to allocate the value of the grand coalition among individual players.
We provide polynomial-size formulations of the least core of these games.
Besides, we give a polynomial-time algorithm for computing the nucleolus, whether or not the core is empty.
This resolves a long-standing gap left by Potters et al. [GEB'06], whose algorithm for the nucleolus of simple flow games with public arcs assumes a non-empty core.
\end{abstract}

\section{Introduction} \label{intro}
Cooperative game theory considers how to distribute the total profit generated by a group of players among its members.
In this paper, we study flow games with public arcs, which model situations in which a commodity, such as goods or information, can flow through a network.
In classical flow games~\cite{KL82}, players control all arcs in the network and cooperate to allow flow from the source to the sink.
However, this assumption fails to capture the widespread existence of shared resources, public goods, and common infrastructure in real-world networks.
To address this limitation, flow games with public arcs were introduced~\cite{KL82,RM96}.
In this generalized model, both private and public arcs exist: private arcs are owned by individual players, while public arcs can be used freely by any coalition.
This extension provides a more realistic framework for modeling situations where players have access to shared resources or infrastructure.
For example, communication networks can feature both private and shared exchange points for Internet service providers, while supply-chain networks may include both public roads and private tollways. 

The core~\cite{GD59} is one of the most attractive solution concepts in cooperative game theory.
An allocation in the core requires that no subset of players has an incentive to deviate from the grand coalition.
However, the core may be empty in many games~\cite{KE03,XI23,AT89}.
The least core, which was introduced by Shapley and Shubik~\cite{SH66}, is a solution concept that relaxes the strict requirements of the core.
More precisely, it consists of allocations that allow coalitions to receive slightly less than their full potential, accepting deviations within a prescribed bound.
This concept has been studied in~\cite{KE03,BY11,FL15}.
The nucleolus, which was introduced by Schmeidler~\cite{SCH69}, is a key solution concept in cooperative game theory.
It is the unique allocation that lexicographically maximizes the vector of deviations sorted in non-decreasing order.
Moreover, it is well known that the nucleolus is always in the least core.
The nucleolus has been widely discussed for a number of cooperative games, including matching games~\cite{KE03,KO20}, spanning tree games~\cite{ME78,GA80}, graphical games~\cite{GR15,SZ17} and arboricity games~\cite{XI23}.

Kalai and Zemel~\cite{KL82,KZ82} introduced flow games as cooperative games arising from network flow problems.
They showed that the core of flow games is always non-empty~\cite{KZ82}.
Deng et al.~\cite{DF09} show that the nucleolus of simple flow games can be computed in polynomial time.
On the other hand, negative results are known when the arc capacities in the network are not all equal to one.
Fang et al.~\cite{FZ02} show that testing membership in the core of a general-capacity flow game is NP-hard.
Deng et al.~\cite{DF09} show that computing the nucleolus is NP-hard for general-capacity flow games.

When there are public arcs in the network, Reijnierse et al.~\cite{RM96} show that the core of a simple flow game is non-empty if and only if there is a minimum cut without public arcs in the network.
In that case, the allocations corresponding to minimum cuts without public arcs belong to the core.
In addition, they characterized other solution concepts for simple flow games, such as the kernel and the kernel-core.
Potters et al.~\cite{PR06} give a polynomial-time algorithm for the nucleolus when core is non-empty.
More recently, Liu et al.~\cite{LX2024} give a polynomial-time algorithm for the approximate core and nucleon.

The main contribution of this work is an efficient characterization of the least core and a polynomial-time algorithm for the nucleolus of flow games with public arcs.
Both results rely on linear programming techniques and their connection to flow games.
These results directly yield efficient algorithms for the related solution concepts.
Moreover, our results on the complexity of computing the nucleolus generalize those of Deng et al.~\cite{DF09} and Potters et al.~\cite{PR06}.
The organization of the paper is as follows.

In Section \ref{sec:pre}, some notions from cooperative game theory and flow games are reviewed.
Section \ref{sec:ls-core} gives a compact formulation of the least core.
Section \ref{sec:nucleolus} gives a polynomial-time algorithm for the nucleolus.
Section \ref{sec:conclusion} gives concluding remarks and possible directions for future research on flow games.

\section{Preliminaries} \label{sec:pre}

\subsection{Cooperative games and solution concepts} \label{sec:coop-game}

A cooperative game $(N,v)$ consists of a \emph{player set} $N=\{1,2,\ldots,n\}$ and a \emph{characteristic function} $v: 2^N \rightarrow \mathbb{R}$ with $v(\emptyset)=0$, where for each \textit{coalition} $S\subseteq N$, $v(S)$ represents the worth of coalition $S$.
For a vector $x$ indexed by a finite ground set $U$ and a subset $S\subseteq U$, write $x(S)=\sum_{i\in S}x(i)$.
For $S\subseteq U$, denote by $\chi_S$ the characteristic vector of $S$, i.e., $\chi_S(e)=1$ if $e\in S$ and $\chi_S(e)=0$ otherwise.
If $S$ consists of a single element $e$, we write $\chi_e$ for $\chi_{\{e\}}$.
A vector $x\in \mathbb{R}^n$ is called an \textit{allocation} of $(N,v)$ if $x(N)=v(N)$.
The \textit{core} of the game $(N,v)$ is defined by:
\begin{equation}
\label{eq:def-core}
\mathcal{C}(N,v)=\{x\in \mathbb{R}^n\mid x(N)=v(N),~x(S)\geq v(S),~\forall S\subset N\},
\end{equation}

Intuitively, an allocation in the core guarantees that each coalition receives at least what it could gain on its own.
We call the game $(N,v)$ \textit{balanced} if the core is non-empty.
Since the core may be empty in many games, the concept of the least core is introduced.
Given a coalition $S\subseteq N$, we call $x(S)-v(S)$ the \textit{excess} of $S$ and denote it by $e(x,S)$.
This concept can be thought of as a measure of the satisfaction of $S$.
The \textit{least core} is the set of allocations that maximize the minimum excess, which can be computed by the following LP:
\begin{maxi}
  {}{\epsilon}
  {\label{lp:ls-core}\tag{$P_1$}}{}
  \addConstraint{x(S)}{\geq v(S)+\epsilon}{\text{~~for all~~}S\subset N}
  \addConstraint{x(N)}{=v(N).}{}
\end{maxi}
Let $\epsilon_1$ be the optimal value of this LP, and let $P_1(\epsilon_1)$ be the set of allocations $x$ such that $(x,\epsilon_1)$ is feasible for ($P_1$).
Now we turn to the concept of the nucleolus.
An allocation $x$ generates a $(2^n-2)$-dimensional excess vector $\theta(x)=(e(x,S_1),\ldots,e(x,S_{2^n-2}))$ that contains all the excess values (except for the excesses of the grand coalition and the empty set) and whose components are arranged in a non-decreasing order, i.e., $e(x,S_1)\leq \ldots \leq e(x,S_{2^n-2})$.
The \textit{nucleolus}, denoted by $\eta(N,v)$, is defined as the unique allocation $x$ that lexicographically maximizes $\theta(x)$.
The nucleolus can be computed by recursively solving a sequence of linear programs $\{(P_r)\}_{r\geq 1}$~\cite{MS79}.
($P_1$) is the least core LP defined above.
For a polytope $P\subset \mathbb{R}^n$, let
$$
\text{Fix}(P)=\{S\subseteq N\mid x(S)=y(S),~\forall x,y\in P\}
$$
denote the set of coalitions fixed for $P$.
In general, given $\epsilon_{r-1}$, we define LP ($P_r$) for $r\geq 2$:
\begin{maxi}
  {}{\epsilon}
  {\label{lp:pr}\tag{$P_r$}}{}
  \addConstraint{x(S)}{\geq v(S)+\epsilon}{\text{~~for all~~}S\notin \text{Fix}(P_{r-1}(\epsilon_{r-1}))}
  \addConstraint{x}{\in P_{r-1}(\epsilon_{r-1}).}{}
\end{maxi}
Let $\epsilon_r$ be the optimal value of ($P_r$), and let $P_r(\epsilon_r)$ be the set of allocations $x$ such that $(x,\epsilon_r)$ is feasible for ($P_r$).
We continue for $r=2,\ldots,n$, or until $P_r(\epsilon_r)$ is a singleton.
Note that the dimension of ($P_r$) decreases by at least one at each step, so it takes at most $n$ steps for $P_r(\epsilon_r)$ to become a singleton.

\subsection{Our method for the nucleolus} \label{sec:gen-nucleolus}

In this subsection, we describe how to compute $P_r(\epsilon_r)$ in detail.
We first solve the separation problem for $P_1(\epsilon_1)$.
For $r\geq 2$, we reduce the computation of $P_r(\epsilon_r)$ to finding a minimum excess coalition whose incidence vector lies outside a given linear subspace.
Note that similar approaches for computing nucleolus have been used for weighted voting games~\cite{PK22} and $b$-matching games~\cite{KO20}.

For a non-empty polytope $P$, let $L(P)=\operatorname{span}\{x-y\mid x,y\in P\}^{\perp}$.
Then $S\in\text{Fix}(P)$ if and only if $\chi_S\in L(P)$, and $\chi_N\in L(P)$ because $x(N)=v(N)$ for every $x\in P$.
For $x\in\mathbb R^N$ and a subspace $L\subseteq\mathbb R^N$, let
\begin{equation}
\Delta(x,L)=\min_{\substack{\emptyset\neq S\subset N\\ \chi_S\notin L}}e(x,S).
\label{eq:Delta-separator}
\end{equation}
Hence $\Delta(x,L(P))$ is the minimum excess over the coalitions not fixed on $P$.
For a polytope $Q\subseteq\mathbb R^d$, a \emph{strong separation oracle} takes $y\in\mathbb Q^d$ and either certifies that $y\in Q$ or returns $(a,b)\in\mathbb Q^d\times\mathbb Q$ such that $a^\top z\leq b<a^\top y$ for every $z\in Q$.
For a positive integer $\varphi$, we say that a polyhedron $Q\subseteq\mathbb R^d$ has \emph{facet complexity at most $\varphi$} if there exists a system of rational linear inequalities with solution set $Q$ in which every inequality has encoding length at most $\varphi$~\cite{GM12}.
A strong separation oracle for $Q\subseteq\mathbb R^d$ is called \emph{polynomial-time} if
given a query $y\in\mathbb{Q}^d$, it runs in time polynomial in
$d$, $\varphi$, and $\langle y\rangle$ where $\varphi$ is the facet complexity of $Q$ and $\langle y\rangle$ is the encoding lengths of $y$.

\begin{lemma}\label{lem:direction-recovery}
Let $P\subseteq\mathbb R^N$ be a non-empty polytope with facet complexity at most $\varphi$.
Given a polynomial-time strong separation oracle for $P$, a rational basis of $L(P)$ can be computed in time polynomial in $|N|$ and $\varphi$.
\end{lemma}

\begin{proof}
We construct a rational basis of $\operatorname{span}\{x-y\mid x,y\in P\}$ incrementally.
Initialize $W=\{0\}$ and maintain $W\subseteq\operatorname{span}\{x-y\mid x,y\in P\}$ by a rational basis.
Given this basis, a rational basis $a_1,\ldots,a_k$ of $W^\perp$ can be computed by Gaussian elimination.
For every $i$, we choose a vertex $x_i^*$ of the maximizing face and a vertex $y_i^*$ of the minimizing face of $a_i^\top x$ over $P$ by successively maximizing the coordinates in a fixed order over each optimal face.
The required linear programs can be solved in time polynomial in $|N|$ and $\varphi$ from the strong separation oracle and the facet-complexity bound~\cite{GM12}.
Moreover, the selected vertices have encoding length polynomially bounded by $|N|$ and $\varphi$.
If $a_i^\top x_i^*>a_i^\top y_i^*$ for some $i$, then $a_i^\top(x_i^*-y_i^*)>0$.
Since $a_i\in W^\perp$, the vector $x_i^*-y_i^*$ lies outside $W$ and is added to its basis.
Thus every iteration increases $\dim W$.
If $a_i^\top x_i^*=a_i^\top y_i^*$ for every $i$, then $a_i^\top(x-y)=0$ for every $i$ and every $x,y\in P$.
Hence every $x-y$ lies in $W$, i.e., $W=\operatorname{span}\{x-y\mid x,y\in P\}$.
Since the construction runs in time polynomial in $|N|$ and $\varphi$, a rational basis of $L(P)=W^\perp$ can therefore be computed within the claimed time.
\end{proof}

The following theorem gives a sufficient condition under which the nucleolus can be computed in polynomial time.

\begin{theorem}\label{thm:Delta-recursion}
Suppose that the following conditions hold:
\begin{enumerate}[(i)]
\item $P_1(\epsilon_1)$ has a polynomial-time strong separation oracle;
\item Every coalition value $v(S)$ is rational and has encoding length at most $\varphi$;
\item Given $x\in P_1(\epsilon_1)\cap\mathbb Q^N$ and a rational basis of a subspace $L\subseteq\mathbb R^N$ containing $\chi_N$, $\Delta(x,L)$ and, when finite, a minimizing coalition can be computed in polynomial time.
\end{enumerate}
Then the nucleolus can be computed in time polynomial in $|N|$ and $\varphi$.
\end{theorem}

\begin{proof}
We first bound the encoding lengths of $\epsilon_1,\ldots,\epsilon_r$.
Fix $r\leq |N|$, replace the objective variable in $(P_j)$ by a separate variable $\xi_j$ for each $1\leq j\leq r$, and collect all constraints in $(P_1),\ldots,(P_r)$.
Since each $\xi_j$ appears only in inequalities of the form $x(S)\geq v(S)+\xi_j$ and is unbounded below, successively solving $(P_1),\ldots,(P_r)$ is equivalent to lexicographically maximizing $(\xi_1,\ldots,\xi_r)$ over this joint system.
Its optimal face is $P_r(\epsilon_r)\times\{(\epsilon_1,\ldots,\epsilon_r)\}$.
Therefore, if $\bar{x}$ is a vertex of $P_r(\epsilon_r)$, then $(\bar{x},\epsilon_1,\ldots,\epsilon_r)$ is a vertex of the joint system.
The joint system has at most $2|N|$ variables, coefficients in $\{-1,0,1\}$, and right-hand sides of encoding length at most $\varphi$ by condition~(ii).
Thus the encoding length of $\epsilon_r$ is polynomially bounded in $|N|$ and $\varphi$ by Cramer's rule.
Moreover, every $P_r(\epsilon_r)$ has facet complexity polynomially bounded in $|N|$ and $\varphi$.

We next prove that a polynomial-time strong separation oracle for $P_r(\epsilon_r)$ and a rational basis of $L(P_r(\epsilon_r))$ can be computed in time polynomial in $|N|$ and $\varphi$.
For $r=1$, condition~(i), the facet-complexity bound of $P_1(\epsilon_1)$, and Lemma~\ref{lem:direction-recovery} give a rational basis of $L(P_1(\epsilon_1))$ within the claimed time.
Suppose that $r\geq2$ and we have an oracle for $P_{r-1}(\epsilon_{r-1})$ and a rational basis of $L_{r-1}=L(P_{r-1}(\epsilon_{r-1}))$.
If $P_{r-1}(\epsilon_{r-1})$ is a singleton, then the recursion terminates.
Otherwise, $L_{r-1}\neq\mathbb R^N$.
Since $\{\chi_e\mid e\in N\}$ spans $\mathbb R^N$, some $\chi_e$ does not belong to $L_{r-1}$, so $\Delta(x,L_{r-1})$ is finite for every $x\in P_{r-1}(\epsilon_{r-1})$.

We next compute $\epsilon_r$.
Since $\chi_S\notin L_{r-1}$ if and only if $S\notin\text{Fix}(P_{r-1}(\epsilon_{r-1}))$, we have
$$
\epsilon_r=\max_{x\in P_{r-1}(\epsilon_{r-1})}\Delta(x,L_{r-1})
$$
and
$$
P_r(\epsilon_r)=\{x\in P_{r-1}(\epsilon_{r-1})\mid\Delta(x,L_{r-1})\geq\epsilon_r\}.
$$
To compute $\epsilon_r$, consider the polyhedron
$$
H_{r-1}=\{(y,\xi)\mid y\in P_{r-1}(\epsilon_{r-1}),~\xi\leq\Delta(y,L_{r-1})\}.
$$
A strong separation oracle for $H_{r-1}$ can be constructed from the oracle for $P_{r-1}(\epsilon_{r-1})$ by induction.
Given a rational query $(x,\eta)$, we first separate $x$ from $P_{r-1}(\epsilon_{r-1})$.
If the oracle returns a separator $(a,b)$, then $((a,0),b)$ separates $(x,\eta)$ from $H_{r-1}$.
Otherwise, $x\in P_{r-1}(\epsilon_{r-1})$ and $\chi_N\in L_{r-1}$, so condition~(iii) computes $\Delta(x,L_{r-1})$ and a minimizing coalition $S$.
If $\eta\leq\Delta(x,L_{r-1})$, then $(x,\eta)\in H_{r-1}$.
Otherwise, the inequality $\xi\leq y(S)-v(S)$ is valid for $H_{r-1}$ and separates $(x,\eta)$.
Notice that $H_{r-1}$ is described by the inequalities defining $P_{r-1}(\epsilon_{r-1})$ together with $\xi-\chi_S^\top y\leq-v(S)$, $\forall \chi_S\notin L_{r-1}$.
By using the bound of facet complexity for $P_{r-1}(\epsilon_{r-1})$ and condition~(ii), the facet complexity of $H_{r-1}$ is polynomially bounded in $|N|$ and $\varphi$.
Therefore we can compute $\epsilon_r$ in time polynomial in $|N|$ and $\varphi$ by using separation--optimization equivalence~\cite{GM12}.

We next construct a strong separation oracle for $P_r(\epsilon_r)$.
Given a rational query $x$, we first separate it from $P_{r-1}(\epsilon_{r-1})$.
If the oracle returns a separator, the returned inequality remains valid for $P_r(\epsilon_r)$.
Otherwise, $x\in P_{r-1}(\epsilon_{r-1})$, so we can compute $\Delta(x,L_{r-1})$ and a minimizing coalition $S$ by condition~(iii).
If $\Delta(x,L_{r-1})\geq\epsilon_r$, then $x\in P_r(\epsilon_r)$.
If $\Delta(x,L_{r-1})<\epsilon_r$, then $-\chi_S^\top y\leq-v(S)-\epsilon_r<-\chi_S^\top x$ for every $y\in P_r(\epsilon_r)$, so $(-\chi_S,-v(S)-\epsilon_r)$ is a strong separator.
The facet-complexity bound and Lemma~\ref{lem:direction-recovery} then give a rational basis of $L(P_r(\epsilon_r))$ within the claimed time.
This completes the induction.

By using separation--optimization equivalence~\cite{GM12}, we can find a vector in $P_r(\epsilon_r)$ in time polynomial in $|N|$ and $\varphi$.
Since the dimension of $P_r(\epsilon_r)$ strictly decreases and $P_r(\epsilon_r)$ becomes a singleton in at most $|N|$ steps, the nucleolus can be computed in time polynomial in $|N|$ and $\varphi$.
\end{proof}

\subsection{Simple flow games}

Consider a directed network $D=(V,E)$ with vertex set $V$ and arc set $E$.
Let $s,t\in V$ be the source and sink vertices of $D$, and assume that each arc in $D$ has a unit flow capacity.
We allow $D$ to have parallel arcs, but no loops.
For any non-empty set $S\subseteq E$, the induced subgraph on $S$, denoted by $D[S]$, is a subgraph of $D$ with arcs in $S$.
Simple flow games study how to allocate the maximum flow value of a simple network among the players.
More precisely, $(N,v)$ is the \textit{simple flow game} defined on a simple network $D=(V,E)$.
The player set of $(N,v)$ consists of a subset of arcs $N\subseteq E$.
An arc in $N$ is called \textit{private}.
Otherwise, it is called \textit{public}.
Let $N^0=E\setminus N$ denote the set of public arcs.
For any coalition $S\subseteq N$,
the characteristic function $v:2^N \rightarrow \mathbb{R}_+$ is defined as the maximum flow value in $G[S\cup N^0]$.
For the function $v$ to be well-defined, i.e., $v(\emptyset)=0$, we assume that $N$ is an $s$-$t$ cut of $D$ in the remainder of this paper.
We further assume that $v(N)\geq 2$ in what follows.
The case where $v(N)=1$ is discussed in Appendix \ref{app:ls-core-nue-v=1}.
In this case, the least core and the nucleolus are reduced to the core and the nucleolus of an associated balanced flow game.

Reijnierse et al.~\cite{RM96} give the following characterization of the core.

\begin{theorem}[Reijnierse et al.~\cite{RM96}] \label{thm:core-cut}
Let $(N,v)$ be a simple flow game defined on $D=(V,E)$ with $N\subseteq E$.
Then $\mathcal{C}(N,v)\neq \emptyset$ if and only if there is a minimum cut without public arcs in $D$.
In this case, $\mathcal{C}(N,v)$ is the convex hull of the characteristic vectors of minimum cuts without public arcs in $D$.
\end{theorem}

\section{A compact formulation for the least core}\label{sec:ls-core}

In this section, we provide a compact characterization of the least core.
The omitted proofs are given in Appendix~\ref{app:prove}.
We first replace the exponentially many coalition constraints with flow constraints and compute a relative interior least core allocation.
We then identify the flows that are tight throughout the least core and encode their common optimality conditions through the sets $M$, $A$, and $R$ and a potential function.
The partition $(M,A,R)$ records the arcs used by all, some, or none of the universal flows.
Complementary slackness then turns this structure into the potential formulation of least core.
We begin with two basic properties of least core allocations.

\begin{lemma} \label{lem:ls-core-non-neg}
$\epsilon_1\leq 0$, and $P_1(\epsilon_1)\subseteq \mathbb{R}^{|N|}_+$.
\end{lemma}

\begin{proof}
See Appendix \ref{app:ls-core-non-neg}.
\end{proof}

This lemma states that the minimum excess is non-positive in any least core allocation.
Define $\mathcal{F}$ as the set of all integer flows of $D$, i.e.,
$$
\mathcal{F}=\{f\mid f\text{ is a flow of }D,~f(e)\in \{0,1\},~\forall e\in E\}.
$$
For $f\in\mathcal{F}$, write $f_N=(f(e))_{e\in N}$ for its restriction to $N$.
For any $x\in \mathbb{R}^{|N|}$ and $f\in \mathcal{F}$, we denote $x(f)=\sum_{e\in N}x(e)f(e)$ and $c(f)=\sum_{e\in E}c(e)f(e)$, where the vector $c\in \mathbb{R}^{|E|}$ is defined as follows:
$$
c(e)=\begin{cases}1,~\text{if }\partial^-e=s, \\
-1,~\text{if }\partial^+e=s, \\
0,~\text{otherwise.}\end{cases}
$$
Clearly, $c(f)$ is the flow value of $f$ for any $f\in \mathcal{F}$.
The following lemma gives a description of the least core that is useful for our purposes.

\begin{lemma} \label{lem:ls-lp-replace}
The inequalities $x(S)\geq v(S)+\epsilon,~\forall S\subset N$ in the linear program ($P_1$) can be replaced by $x(f)\geq c(f)+\epsilon$, $\forall f\in \mathcal{F}$, and $x(e)\geq 0$, $\forall e\in N$.
\end{lemma}

\begin{proof}
See Appendix \ref{app:ls-lp-replace}.
\end{proof}

Since the inequalities in Lemma \ref{lem:ls-lp-replace} can be separated efficiently by minimum cost flow techniques, we have the following lemma.

\begin{lemma}\label{lem:ls-compute}
An allocation in $P_1(\epsilon_1)$ can be found in polynomial time.
\end{lemma}

\begin{proof}
Based on the polynomial equivalence of optimization and separation problems (see, e.g., \cite{GM12}), it suffices to show that given $x$ and $\epsilon$, we can efficiently check whether $x(f)\geq c(f)+\epsilon$ holds for any $f\in \mathcal{F}$.
This can be checked by solving a minimum cost flow problem on $D$ with respect to the arc costs $w(e)=x(e)-c(e)$ for any $e\in E$.
The value of a minimum cost flow of $D$ can be computed in polynomial time (see, e.g., \cite{AM93}).
\end{proof}

Following the standard implicit-equality characterization of the relative interior of a polyhedron (see, e.g., \cite{schrijver1986theory}), we use the following characterization of the relative interior.
Let $P=\{x\mid Ax\leq b\}$ be a non-empty polytope, and define $I_<$ as the set of indices $i$ for which there exists a point $y \in P$ satisfying $A_i y < b_i$.
A point $x \in P$ is a \textit{relative interior point} of $P$ if $A_i x < b_i$ for every $i\in I_<$.
A flow $f\in \mathcal{F}$ is called $x$-\textit{tight} for a given $x\in P_1(\epsilon_1)$ if $x(f)=c(f)+\epsilon_1$.
A flow is a \textit{universal flow} if it is $x$-tight for all $x\in P_1(\epsilon_1)$.
Let $\mathcal{F}^x$ denote the set of all $x$-tight flows in $\mathcal{F}$.

\begin{lemma} \label{lem:inter-point}
A relative interior point in $P_1(\epsilon_1)$ can be computed in polynomial time.
\end{lemma}

\begin{proof}
We give a polynomial-time combinatorial algorithm for finding a relative interior point.
Given an initial least core allocation $x$, the high-level idea of this algorithm is to attempt to find a flow $f$ that is an $x$-tight flow but not a universal flow.
If such a flow exists, $x$ is updated to increase the allocation on $f$; otherwise, $x$ is output as a relative interior point of the least core.
The detailed steps and complexity of this algorithm are presented in Appendix \ref{app:inter-point}.
\end{proof}

Fix a relative interior point $x^*$ in $P_1(\epsilon_1)$.
Clearly, for any relative interior point $x^*$ in $P_1(\epsilon_1)$, the set of universal flows is equal to $\mathcal{F}^{x^*}$.
Conversely, if $f\in\mathcal{F}$ is not a universal flow, then the inequality $x(f)\geq c(f)+\epsilon_1$ is strict for any relative interior point $x$ in $P_1(\epsilon_1)$.
Thus, $\mathcal{F}^{x^*}$ does not depend on the choice of $x^*$.
Note that a flow $f$ is an $x$-tight flow if and only if it is a solution to the minimum cost flow problem with costs given by $x-c$.
Specifically, it can be obtained by solving the following linear program:
\begin{mini}
{}{\sum_{e\in N}x(e)f(e)-\sum_{e\in E}c(e)f(e)} {\label{lp:mincost-flow}}{} \addConstraint{\sum_{e:\partial^-e=v}f(e)-\sum_{e:\partial^+e=v}f(e)}{= 0}{\text{~~for all~~}v\in V\setminus \{s,t\}} \addConstraint{0\leq f(e)}{\leq 1}{\text{~~for all~~} e\in E.}
\end{mini}
By setting $\pi(s)=1$ and $\pi(t)=0$, we obtain the dual of the linear program \eqref{lp:mincost-flow} as follows.
\begin{maxi!}
{}{\sum_{e\in E}y(e)} {\label{lp:mincost-dual}}{} \addConstraint{\pi(s)=1,\pi(t)}{=0}{} \addConstraint{\pi(\partial^-e)-\pi(\partial^+e)+y(e)}{\leq x(e)\protect\label{cons-pot-N}}{\text{~~for all~~}e\in N} \addConstraint{\pi(\partial^-e)-\pi(\partial^+e)+y(e)}{\leq 0\protect\label{cons-pot-N0}}{\text{~~for all~~}e\in N^0} \addConstraint{y(e)}{\leq 0}{\text{~~for all~~}e\in E.}
\end{maxi!}

Note that the optimal solutions to the linear program \eqref{lp:mincost-flow} always lie on the optimal face of the flow polytope corresponding to the linear objective with weights $x-c$.
Thus, we define the following partition $(M_x,A_x,R_x)$ of $E$: $M_x=\{e\mid f(e)=1,~\forall f\in\mathcal{F}^x\}$, $R_x=\{e\mid f(e)=0,~\forall f\in\mathcal{F}^x\}$ and $A_x=E\setminus(M_x\cup R_x)$ which consists of the remaining arcs.
Note that $\mathcal{F}^x$ is precisely the intersection of $\mathcal{F}$ and the set of optimal solutions to the linear program \eqref{lp:mincost-flow}.
We can characterize $\mathcal{F}^x$ in the following lemma.
\begin{lemma}\label{lem:opt-flow}
$\mathcal{F}^x=\{f\in\mathcal{F}\mid f(e)=1,~\forall e\in M_x,~f(e)=0,~\forall e\in R_x\}$.
\end{lemma}
\begin{proof}
See Appendix \ref{app:opt-flow}.
\end{proof}
Moreover, we show that these three sets can be computed efficiently, and give a combinatorial algorithm as follows.

\begin{lemma} \label{lem:MAR-x}
Given $x\in P_1(\epsilon_1)$, the sets $M_x$, $A_x$, and $R_x$ can be identified in polynomial time.
\end{lemma}

\begin{proof}
See Appendix~\ref{app:MAR-x}.
\end{proof}

Fix a universal flow $f^*\in\mathcal{F}^{x^*}$.
The following lemma shows that the least core can be expressed as the set of vectors for which the universal flow remains a minimum cost flow.

\begin{lemma}\label{lem:ls-core-sen-any}
$P_1(\epsilon_1)=\{x\in \mathbb{R}^{|N|}_+\mid \mathcal{F}^{x^*}\subseteq \mathcal{F}^x,~x(N)=v(N),~x(f^*)=c(f^*)+\epsilon_1\}$.
\end{lemma}

\begin{proof}
Let $x\in P_1(\epsilon_1)$.
If $f\in \mathcal{F}^{x^*}$, then $f$ is an $x$-tight flow.
Therefore, $\mathcal{F}^{x^*}\subseteq \mathcal{F}^x$.
On the other hand, let $x\in \mathbb{R}^{|N|}_+$ satisfy $\mathcal{F}^{x^*}\subseteq \mathcal{F}^x$, $x(N)=v(N)$ and $x(f^*)=c(f^*)+\epsilon_1$.
Since $f^*$ is a minimum cost flow with $x-c$ as costs, we have $x(f)-c(f)\geq x(f^*)-c(f^*)=\epsilon_1$ for any $f\in \mathcal{F}$.
By Lemma \ref{lem:ls-lp-replace}, we have $x\in P_1(\epsilon_1)$.
\end{proof}

Lemma \ref{lem:ls-core-sen-any} states that the problem of characterizing the least core can be reduced to addressing the constraint $\mathcal{F}^{x^*}\subseteq \mathcal{F}^x$.
Let $M=M_{x^*}$, $A=A_{x^*}$ and $R=R_{x^*}$.
By Lemma \ref{lem:inter-point} and Lemma \ref{lem:MAR-x}, it is easy to see that these three sets can be computed in polynomial time.
For example, public arcs in $D$ are indicated by dashed lines in Figure \ref{fig:ex}.
A relative interior point of the least core is given by $x^*((s,u))=x^*((s,w))=\frac{1}{2}$, $x^*((s,t))=1$ and $x^*((w,v))=0$.
Moreover, the corresponding partition is $M=\{(w,t)\}$, $A=\{(s,u),(s,w),(s,t),(u,w)\}$, and $R=\{(w,v),(v,s)\}$.
\begin{figure}[H]
    \centering
    \includegraphics[width=8cm]{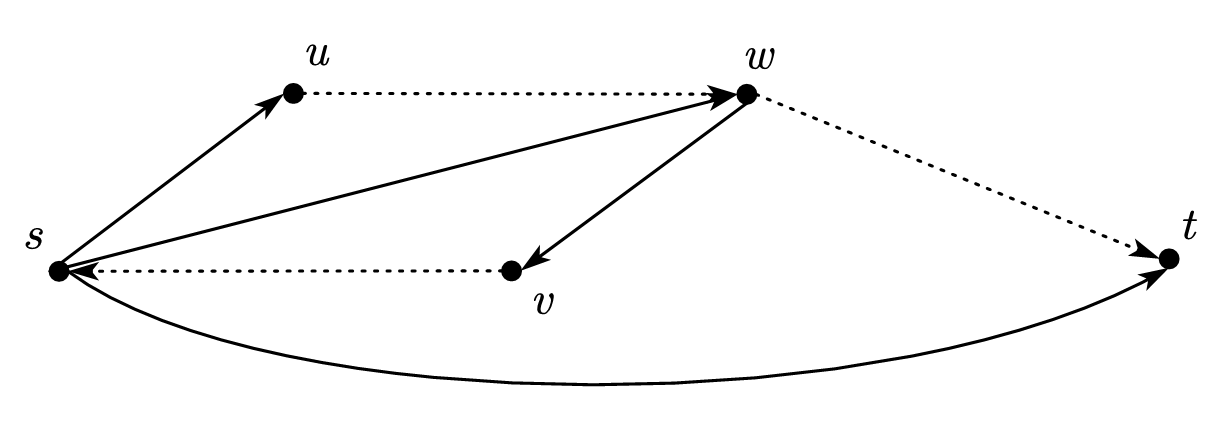}
    \caption{An example of arcs in $M$, $A$ and $R$.}
    \label{fig:ex}
\end{figure}

Define $x_\pi(e)$ as the \textit{reduced cost} of arc $e$, given by $x_\pi(e)=x(e)+\pi(\partial^+e)-\pi(\partial^-e)$, $\forall e\in N$ and $x_\pi(e)=\pi(\partial^+e)-\pi(\partial^-e)$, $\forall e\in N^0$.
We next consider the condition $\mathcal{F}^{x^*}\subseteq \mathcal{F}^x$.
This condition can be expressed as the complementary slackness conditions of linear programming \eqref{lp:mincost-flow}.

\begin{lemma} \label{lem:sen-any}
Let $x\in \mathbb{R}^{|N|}$.
The following statements are equivalent:
\begin{enumerate}[(i)]
\item $\mathcal{F}^{x^*}\subseteq \mathcal{F}^x$;
\item $M_x\subseteq M$, $A\subseteq A_x$, $R_x\subseteq R$;
\item There exists $\pi\in \mathbb{R}^{|V|}$ such that
$$
\begin{aligned}&\pi(s)=1,\pi(t)=0,\\
&x_\pi(e)\leq 0,\forall e\in M,\\
&x_\pi(e)= 0,\forall e\in A,\\
&x_\pi(e)\geq 0,\forall e\in R.\\
\end{aligned}
$$
\end{enumerate}
\end{lemma}
\begin{proof}
$(i)\Rightarrow (ii)$.
Assume $\mathcal{F}^{x^*}\subseteq \mathcal{F}^x$, and let $f\in \mathcal{F}^{x^*}$.
If $e\in M_x$, then $f(e)=1$.
Therefore, $e\in M$.
If $e\in R_x$, then $f(e)=0$.
Therefore, $e\in R$.
Since $M\cup A \cup R=M_x\cup A_x\cup R_x=E$, we have $A\subseteq A_x$.

$(ii)\Rightarrow (iii)$.
Let $(y,\pi)$ be a dual optimal solution of the linear program \eqref{lp:mincost-flow} whose objective function assigns weights $x-c$.
Then $\pi$ satisfies $\pi(s)=1$ and $\pi(t)=0$.
We can write constraints \eqref{cons-pot-N} and \eqref{cons-pot-N0} as $y(e)\leq x_\pi(e)$, $\forall e\in E$.
Since the coefficient of each variable $y(e)$ in the objective function is $1$, we can choose an optimal solution with $y(e)=\min\{0,x_\pi(e)\}$.
By condition (ii), if $e\in M\cup A$, then $e\in M_x\cup A_x$.
By definition of $M_x$ and $A_x$, there exists $f\in \mathcal{F}^x$ such that $f(e)=1$.
By the complementary slackness conditions, we have $x_\pi(e)-y(e)=0$.
Therefore, $x_\pi(e)=y(e)\leq 0$.
By condition (ii), if $e\in A\cup R$, then $e\in A_x\cup R_x$.
By the definition of $A_x$ and $R_x$, there exists $f\in \mathcal{F}^x$ such that $f(e)=0$.
By the complementary slackness conditions, we have $y(e)=0$.
Therefore, $x_\pi(e)\geq y(e)=0$.

$(iii)\Rightarrow (i)$.
Let $\pi$ satisfy the conditions in (iii), and let $f\in \mathcal{F}^{x^*}$.
We prove that $f\in \mathcal{F}^x$.
Let $y\in\mathbb{R}^{|E|}$ be defined by $y(e)=\min\{0,x_\pi(e)\}$.
If $f(e)=1$, then $e\in M\cup A$.
Therefore, $y(e)-x_\pi(e)=0$.
If $f(e)=0$, then $e\in A\cup R$.
Therefore, $y(e)=0$.
Thus, $f$ and $(\pi,y)$ satisfy the complementary slackness conditions, i.e., $f\in \mathcal{F}^x$.
\end{proof}

We next prove the structural result on allocations in the least core.

\begin{lemma} \label{lem:lscore-alloc}
Let $x\in P_1(\epsilon_1)$.
We have
\begin{enumerate}[(i)]
\item $x(e)=0$, $\forall e\in R\cap N$;
\item $M\subseteq N^0$.
\end{enumerate}
\end{lemma}

\begin{proof}
To prove (i), we assume there exists $e_1\in R\cap N$ such that $x(e_1)>0$.
We first prove that $N\setminus R\neq \emptyset$.
Otherwise, $N\subseteq R$.
Therefore, if $f\in \mathcal{F}^{x^*}$, then $f(e)=0$, $\forall e\in N$.
Thus, zero flow is a universal flow and $\epsilon_1=0$.
However, if $f^*$ is a maximum flow of $D$, we have $v(N)=x(N)\geq x(f^*)\geq c(f^*)=v(N)$, i.e., $f^*$ is a universal flow.
Since $N$ is an $s$-$t$ cut and $c(f^*)=v(N)>0$, $f^*$ uses some arc in $N$.
This contradicts $N\subseteq R$.
We next define $\delta(x,e)$ as the amount by which we transfer the allocation of $x$ on $e$ to $N\setminus R$ while maintaining the least core allocation.
If there exists $f\in \mathcal{F}\setminus \mathcal{F}^{x^*}$ such that $f(e)=1$, then we define
\begin{equation}\label{eq:delta}
\delta(x,e)=\min\{x(e),\min\{x(f)-c(f)\mid f\in \mathcal{F}\setminus \mathcal{F}^{x^*},~f(e)=1\}-\epsilon_1\},
\end{equation}
otherwise, we define $\delta(x,e)=x(e)$.
Let $\bar{x}=\frac{1}{2}(x+x^*)$.
Then $\bar{x}\in P_1(\epsilon_1)$.
Next, we show that $\delta(\bar{x},e_1)>0$.
If there is no $f \in \mathcal{F}\setminus\mathcal{F}^{x^*}$ such that $f(e_1)=1$, then, by definition, $\delta(\bar{x},e_1)=\bar{x}(e_1)>0$.
Otherwise, for every $f\in\mathcal{F}\setminus\mathcal{F}^{x^*}$ with $f(e_1)=1$, since $f\notin\mathcal{F}^{x^*}$, we have $x^*(f)-c(f)>\epsilon_1$.
On the other hand, since $x\in P_1(\epsilon_1)$, Lemma~\ref{lem:ls-lp-replace} implies that $x(f)-c(f)\geq \epsilon_1$.
Hence, $\bar{x}(f)-c(f)>\epsilon_1$.
Since $\mathcal{F}$ is finite, the minimum over all such flows is also strictly larger than $\epsilon_1$.
Together with $\bar{x}(e_1)>0$, this yields $\delta(\bar{x},e_1)>0$.

Let $x'=\bar{x}+\delta(\bar{x},e_1)(\frac{1}{|N|-1}\chi_{N\setminus \{e_1\}}-\chi_{e_1})$.
We prove that $x'\in P_1(\epsilon_1)$.
Clearly, $x'(e)\geq 0$ for all $e\in N\setminus\{e_1\}$, $x'(e_1)=\bar{x}(e_1)-\delta(\bar{x},e_1)\geq 0$ and $x'(N)=v(N)$.
For any $f\in \mathcal{F}$ with $f(e_1)=1$, we have $x'(f)-c(f)\geq\bar{x}(f)-c(f)-\delta(\bar{x},e_1)\geq \epsilon_1$.
Otherwise, $x'(f)-c(f)\geq \bar{x}(f)-c(f)\geq \epsilon_1$.
By Lemma~\ref{lem:ls-lp-replace}, $x'\in P_1(\epsilon_1)$.
Choose $e_0\in N\setminus R$ and a universal flow $f$ with $f(e_0)=1$.
We have $x'(f)-c(f)>\bar{x}(f)-c(f)=\epsilon_1$, which contradicts the fact that $f$ is a universal flow.

Next, we prove (ii).
We assume that $M\cap N\neq \emptyset$.
We first prove that $x(e)=0$, $\forall e\in A\cap N$.
Otherwise, assume there exists $e_1\in A\cap N$ such that $x(e_1)>0$.
Let $\delta(x,e)$ be defined as in \eqref{eq:delta}, $\bar{x}=\frac{1}{2}(x+x^*)$, and $x'=\bar{x}+\delta(\bar{x},e_1)(\frac{1}{|M\cap N|}\chi_{M\cap N}-\chi_{e_1})$.
As in the proof of (i), we have $x'\in P_1(\epsilon_1)$.
Since $e_1\in A$, there exists a universal flow $f\in\mathcal F^{x^*}$ such that $f(e_1)=0$.
Thus, $x'(f)-c(f)>\bar{x}(f)-c(f)=\epsilon_1$, which contradicts the fact that $f$ is a universal flow.
Therefore, we have $x(e)=0$, $\forall e\in A\cap N$.

Let $f\in\mathcal{F}^{x^*}$.
According to the flow decomposition theorem, there exist arc-disjoint $(s,t)$-paths $P_1,\ldots,P_{k}$ and directed cycles $C_1,\ldots,C_{k'}$ in $D$ such that $f=\sum_{i=1}^k\chi_{P_i}+\sum_{j=1}^{k'}\chi_{C_j}$.
If $k=0$, then $\epsilon_1=c(f)+\epsilon_1=x^*(f)=x^*(M\cap N)>0$, which contradicts Lemma~\ref{lem:ls-core-non-neg}.
Therefore, $k\geq 1$.
If there is some $P_i$ such that $P_i\cap M=\emptyset$, then $f-\chi_{P_i}$ remains a universal flow by using Lemma~\ref{lem:opt-flow}.
Since no universal flow is a circulation, this deletion process cannot remove all $s$-$t$ paths.
Therefore, we assume $P_i\cap M\neq \emptyset$ for $i=1,\ldots,k$.
If there exists $P_{i_0}$ such that $x^*(P_{i_0})>1$, then $x^*(f)-c(f)=x^*(f-\chi_{P_{i_0}})-c(f-\chi_{P_{i_0}})+x^*(P_{i_0})-1>\epsilon_1$, which contradicts $f\in\mathcal{F}^{x^*}$.
Therefore, $x^*(P_i)\leq 1$ for $i=1,\ldots,k$.
Let $f'=\sum_{i=1}^k\chi_{P_i}$.
Since $\epsilon_1\leq x^*(f')-c(f')\leq x^*(f)-c(f)=\epsilon_1$, we have $\sum_{j=1}^{k'}x^*(C_j)=x^*(f)-x^*(f')=0$.
Besides, $x^*(f')-c(f')=\epsilon_1$, i.e., $f'$ is a universal flow.
By (i), $x^*(e)=0$ for all $e\in N\setminus M$.
Therefore, we have
$$
k\leq v(N)=x^*(N)=x^*(M\cap N)=\sum_{i=1}^kx^*(P_i)\leq k.
$$
Comparing the first and last terms of the above chain, we obtain $v(N)=x^*(M\cap N)=k$, i.e., $f'$ is a maximum flow of $D$.
Besides, $x^*(P_i)=1$ and $\epsilon_1=x^*(f')-c(f')=x^*(f')-k=0$.
According to the assumption, $c(f')\geq 2$.
Since $M\cap P_i\neq \emptyset$, it follows that $M\setminus P_i\neq \emptyset$.
Thus, the path flow $\chi_{P_i}$ is not a universal flow.
Thus, we have $0=x^*(P_i)-1>\epsilon_1=0$, a contradiction.
\end{proof}

We next define the following set of potential functions:
$$
\Pi=\left \{\pi\in \mathbb{R}^{|V|}~\middle|~\begin{aligned}&\pi(s)=1,~\pi(t)=0,\\
&\text{diff}(\pi,e)\geq 0,\forall e\in M\cup (A\cap N),\\
&\text{diff}(\pi,e)=0,\forall e\in A\cap N^0,\\
&\text{diff}(\pi,e)\leq 0,\forall e\in R,\\
&\sum_{e\in M}\text{diff}(\pi,e)=-\epsilon_1,\\
&\sum_{e\in A}\text{diff}(\pi,e)=v(N),\\
&\sum_{e\in E}\text{diff}(\pi,e)f^*(e)=c(f^*)\end{aligned}\right \}.
$$
The following theorem states that each potential function in $\Pi$ determines an allocation in the least core.
The proof mainly combines Lemma \ref{lem:ls-core-sen-any} and Lemma \ref{lem:sen-any}.
Since $\Pi$ contains $O(|E|)$ constraints, we have shown that the least core can be characterized by a linear number of constraints.

\begin{theorem} \label{thm:ls-core-des-v>1}
$$
P_1(\epsilon_1)=\left \{x\in\mathbb{R}^{|N|}~\middle|~x(e)=\textup{diff}(\pi,e),~\forall e\in A\cap N,~x(e)=0,~\forall e\in N\setminus A,~\pi\in\Pi\right \}.
$$
\end{theorem}

\begin{proof}
Let $P$ denote the set on the right-hand side of the equality, and let $x\in P_1(\epsilon_1)$.
We prove that $x\in P$.
From Lemma \ref{lem:lscore-alloc}, $x(e)=0$ for all $e\in N\setminus A$.
Next, we prove that there exists $\pi\in \Pi$ such that $x(e)=\text{diff}(\pi,e)$ for all $e\in A\cap N$.
By Lemma \ref{lem:sen-any}, there exists $\pi\in\mathbb{R}^{|V|}$ with $\pi(s)=1$ and $\pi(t)=0$.
For $e\in A\cap N$, we have $\text{diff}(\pi,e)=x(e)\ge 0$.

It remains to show that $\pi\in\Pi$.
For $e\in M$, Lemma \ref{lem:sen-any} gives $\text{diff}(\pi,e)\ge 0$.
For any $e\in A\cap N^0$, we have $\text{diff}(\pi,e)=0$.
For any $e\in R$, we have $\text{diff}(\pi,e)\leq 0$.
Therefore $\sum_{e\in A}\text{diff}(\pi,e)=x(A\cap N)=x(N)=v(N)$.
Let $y(e)=\min\{0,x_\pi(e)\}$.
According to the complementary slackness conditions, $(\pi,y)$ is a dual optimal solution of the linear program \eqref{lp:mincost-flow} with $x-c$ as arc costs.
Therefore, $\sum_{e\in M}\text{diff}(\pi,e)=-y(E)=-\epsilon_1$.
Since $f^*$ is a universal flow, we have $x(f^*)=c(f^*)+\epsilon_1$.
Thus, $\sum_{e\in E}\text{diff}(\pi,e)f^*(e)=\sum_{e\in M}\text{diff}(\pi,e)+x(f^*)=c(f^*)$.
Therefore, $\pi\in \Pi$, and $x(e)=\text{diff}(\pi,e)$ for all $e\in A\cap N$.

On the other hand, let $x\in P$.
We prove that $x\in P_1(\epsilon_1)$.
Let $\pi\in \Pi$ satisfy $x(e)=\text{diff}(\pi,e)$, $\forall e\in A\cap N$.
First, $x(N)=\sum_{e\in A\cap N}\text{diff}(\pi,e)=\sum_{e\in A}\text{diff}(\pi,e)=v(N)$.
Moreover, $x(f^*)=\sum_{e\in A}\text{diff}(\pi,e)f^*(e)=\sum_{e\in E}\text{diff}(\pi,e)f^*(e)-\sum_{e\in M}\text{diff}(\pi,e)f^*(e)=c(f^*)+\epsilon_1$.
Next, we prove that $\mathcal{F}^{x^*}\subseteq \mathcal{F}^x$.
For any $e\in M$, we have $x_\pi(e)=\pi(\partial^+e)-\pi(\partial^-e)\leq 0$.
For any $e\in A\cap N$, we have $x_\pi(e)=\pi(\partial^+e)-\pi(\partial^-e)+x(e)=0$.
For any $e\in A\cap N^0$, we have $x_\pi(e)=\pi(\partial^+e)-\pi(\partial^-e)=0$.
For any $e\in R$, we have $x_\pi(e)=\pi(\partial^+e)-\pi(\partial^-e)\geq 0$.
By Lemma \ref{lem:sen-any}, $\mathcal{F}^{x^*}\subseteq \mathcal{F}^x$.
By Lemma \ref{lem:ls-core-sen-any}, we obtain $x\in P_1(\epsilon_1)$.
\end{proof}

\section{A polynomial-time algorithm for the nucleolus}\label{sec:nucleolus}

In this section, we prove that the nucleolus of a simple flow game can be computed in polynomial time by using the method in Subsection~\ref{sec:gen-nucleolus}.
We prove the following theorem in two steps.
Subsection~\ref{subsec:L-flow} establishes a polynomial-time algorithm for the minimum cost $L$-nontrivial flow problem, which plays a key role in computing the nucleolus.
Subsection~\ref{subsec:delta} uses this algorithm to compute $\Delta(x,L)$ and completes the proof of Theorem~\ref{thm:nucleolus-general}.

\begin{theorem}\label{thm:nucleolus-general}
The nucleolus of a simple flow game can be computed in polynomial time.
\end{theorem}

For $x\in\mathbb{R}_+^N$ and a subspace $L\subseteq\mathbb R^N$, let
\begin{equation}
A(x,L)=\min\{x(f)-c(f)\mid f\in\mathcal F,~f_N\notin L\},
\label{eq:A-separator}
\end{equation}
and
\begin{equation}
B(x,L)=\min\{x(e)+x(f)-c(f)\mid e\in N,~\chi_e\notin L,~f\in\mathcal F,~f(e)=0\}.
\label{eq:B-separator}
\end{equation}
Intuitively, if a coalition attaining $\Delta(x,L)$ coincides with the private arc support of its maximum flow, then $\Delta(x,L)$ can be computed using $A(x,L)$.
Otherwise, the nonnegativity of $x$ allows an optimal coalition to be chosen as the union of such a support and one additional arc $e$ with $\chi_e\notin L$, which is captured by $B(x,L)$.
Thus, computing $\Delta(x,L)$ reduces to evaluating $A(x,L)$ and $B(x,L)$ and taking their minimum.
Thus we have the following lemma.

\begin{lemma}\label{lem:Delta-decomposition}
For every $x\in\mathbb R_+^N$ and every subspace $L\subseteq\mathbb R^N$ containing $\chi_N$,
\begin{equation}
\Delta(x,L)=\min\{A(x,L),B(x,L)\}.
\label{eq:exact-separator}
\end{equation}
\end{lemma}

\begin{proof}
We first prove $\min\{A(x,L),B(x,L)\}\leq\Delta(x,L)$.
Fix $\emptyset\neq S\subset N$ with $\chi_S\notin L$, and let $f$ be a maximum flow in $G[S\cup N^0]$.
Let $U=\{e\in N\mid f(e)=1\}$.
If $f_N\notin L$, then
$$
A(x,L)\leq x(f)-c(f)\leq x(S)-v(S)=e(x,S).
$$
If $f_N\in L$, then $\chi_{S\setminus U}=\chi_S-f_N\notin L$, so $\chi_e\notin L$ for some $e\in S\setminus U$.
Since $f(e)=0$, we have
$$
B(x,L)\leq x(e)+x(f)-c(f)\leq x(S\setminus U)+x(U)-v(S)=e(x,S).
$$
Minimizing over all admissible $S$ gives $\min\{A(x,L),B(x,L)\}\leq\Delta(x,L)$.

On the other hand, we prove that $\Delta(x,L)\leq A(x,L)$ when $A(x,L)\leq B(x,L)$.
Let $f\in\mathcal{F}$ attain the minimum in~\eqref{eq:A-separator} and $U=\{e\in N\mid f(e)=1\}$.
Since $f_N\notin L$ and $\chi_N\in L$, we have $\emptyset\neq U\subset N$.
Consequently,
$$
\Delta(x,L)\leq x(f)-c(f)=A(x,L).
$$

It remains to prove $\Delta(x,L)\leq B(x,L)$ when $B(x,L)<A(x,L)$.
Let $(e,f)$ attain the minimum in~\eqref{eq:B-separator} where $e\in N$ and $f\in\mathcal{F}$, and let $U=\{e\in N\mid f(e)=1\}$.
We prove by contradiction that $f_N\in L$.
If $f_N\notin L$, then
$$
A(x,L)\leq x(f)-c(f)\leq x(f)-c(f)+x(e)=B(x,L),
$$
contrary to $B(x,L)<A(x,L)$.
Since $f_N\in L$ and $\chi_e\notin L$, we have $\chi_{U\cup\{e\}}=f_N+\chi_e\notin L$.
This implies $\emptyset\neq U\cup\{e\}\subset N$.
Therefore,
$$
\Delta(x,L)\leq e(x,U\cup\{e\})\leq x(U)+x(e)-c(f)=B(x,L).
$$
Combining the inequalities in the two directions above, we obtain \eqref{eq:exact-separator} which proves the lemma.
\end{proof}
We then introduce the minimum cost $L$-nontrivial flow problem to compute $A(x,L)$ in the next subsection.

\subsection{The minimum cost $L$-nontrivial flow problem}\label{subsec:L-flow}

Consider a directed network $D=(V,E)$ with vertex set $V$ and arc set $E$.
Let $s,t\in V$ be the source and sink vertices of $D$, and assume that each arc in $D$ has a unit flow capacity.
Let $w:E\to\mathbb Q$ be an arc cost function, and let $L\subseteq\mathbb R^E$ be a linear subspace given by a rational basis.
For a flow $f\in\mathcal F$, define its cost by $w(f)=\sum_{e\in E}w(e)f(e)$.
We call $f$ $L$-trivial if $f\in L$ and $L$-nontrivial otherwise.
The minimum cost $L$-nontrivial flow problem is to find an $L$-nontrivial flow $f$ that minimizes $w(f)$ or to determine that no such flow exists.

We next design a polynomial-time algorithm for the minimum cost $L$-nontrivial flow problem.
Our algorithm starts with a minimum cost flow in $D$ and searches for a minimum cost augmentation along a cycle or path in its residual network that makes the flow $L$-nontrivial, or determines that no $L$-nontrivial flow exists.

Let $f^0$ be a minimum cost flow of $D$ in $\mathcal F$.
The residual network $R=(V,E')$ contains a forward arc for every unused arc $e\in E$, with cost $w(e)$, and a reverse arc for every used arc $e\in E$, with cost $-w(e)$.
For each $e\in E'$, write its cost as $w'(e)$.
For a residual walk $W$, let $\chi^W\in\mathbb R^E$ be its signed incidence vector, whose $e$-coordinate is the number of traversals of the forward copy of $e$ minus the number of traversals of its reverse copy.
Let $\mathcal W(f^0)$ be the collection of all cycles, $s$-$t$ paths and $t$-$s$ paths in $R$.
We first establish a simple but crucial lemma that characterizes when augmentation yields an $L$-nontrivial flow.

\begin{lemma}\label{lem:augment-L-nontrivial}
Suppose that $f^0\in L$, and let $f'$ be the flow obtained by augmenting $f^0$ along $W\in\mathcal W(f^0)$.
Then $f'$ is $L$-nontrivial if and only if $\chi^W\notin L$.
\end{lemma}

\begin{proof}
Augmenting along $W$ changes the flow by $\chi^W$, so $f'=f^0+\chi^W$.
Since $f^0\in L$ and $L$ is a linear subspace, $f'\notin L$ if and only if $\chi^W\notin L$.
\end{proof}

The next lemma gives a polynomial-time algorithm for finding a minimum cost residual walk satisfying the condition in Lemma~\ref{lem:augment-L-nontrivial}.

\begin{lemma}\label{lem:residual-L-flow}
Given a minimum cost flow $f^0\in\mathcal{F}$, a minimum cost $W\in\mathcal W(f^0)$ with $\chi^W\notin L$ or a certificate that no such $W$ exists can be found in polynomial time.
\end{lemma}

\begin{proof}
If $L=\mathbb R^E$, then no such $W$ exists.
Otherwise, choose a rational matrix $H$ with polynomially encoded entries such that $L=\ker H$, and let $h$ be the number of rows of $H$.
After multiplying the rows of $H$ by suitable positive integers, we may assume that $H$ has integer entries.
For every residual arc $e\in E'$ corresponding to an original arc $a\in E$, set $g(e)=H_{\cdot,a}$ if $e$ is forward and $g(e)=-H_{\cdot,a}$ if $e$ is reverse.
For any residual walk $W$, define $g(W)$ as the sum of $g(e)$ over its arc traversals.
Linearity gives $g(W)=H\chi^W$.
Consequently, $\chi^W\notin L$ if and only if $g(W)\neq0$.

We next construct a new cost function that is nonnegative on every residual arc.
Since $f^0$ is a minimum cost flow, $R$ has no cycle with negative cost.
We add a new vertex $s_0$ and an arc $(s_0,v)$ with zero cost for every $v\in V$.
Let $d(v)$ be the cost of a shortest path from $s_0$ to $v$ in the resulting graph.
By the definition of the function $d$, we have $d(v)\le d(u)+w'(e)$ for each $e=(u,v)\in E'$.
Define the cost $\bar w(e)=w'(e)+d(u)-d(v)$ for each $e=(u,v)\in E'$.
Then $\bar w(e)\ge0$ for every $e\in E'$.
Moreover, for every walk $Q$ from $x$ to $y$ in $R$, we have
\begin{equation}\label{eq:re-weight}
\bar w(Q)=w'(Q)+d(x)-d(y).
\end{equation}

Let $G=\max_{e,i}|g_i(e)|$ and $Z=|V|G$.
Every residual cycle or path $W$ has at most $|V|$ arcs, so each component of $g(W)$ has absolute value at most $Z$.
If $Z=0$, then $g(e)=0$ for every $e\in E'$, and hence no $W$ with $\chi^W\notin L$ exists.
Otherwise, let $p_1,\ldots,p_K$ be the first $K=\lceil\log_2(Z+1)\rceil$ primes.
For each row $i$ of $H$ and each prime $p_j$, we construct the product graph with vertex set $V\times\mathbb Z_{p_j}$.
Each vertex in the product graph has two coordinates and is denoted by $\pstate{v}{r}$, where $v\in V$ and $r\in \mathbb{Z}_{p_j}$.
A residual arc $e=(u,v)$ induces the arc from $\pstate{u}{r}$ to $\pstate{v}{(r+g_i(e))\bmod p_j}$ with cost $\bar w(e)$ for every $r\in\mathbb Z_{p_j}$.

We next compute shortest paths from $\pstate{v}{0}$ to $\pstate{v}{r}$ for every $v\in V$ and $r\neq0$, from $\pstate{s}{0}$ to $\pstate{t}{r}$, and from $\pstate{t}{0}$ to $\pstate{s}{r}$.
Consider any path $\pstate{v_0}{r_0}\pstate{v_1}{r_1}\ldots\pstate{v_\ell}{r_\ell}$ where $r_0=0$ and $r_\ell=r$.
Let $e_k=(v_{k-1},v_k)$ be the residual arc corresponding to the $k$-th product graph arc.
By the construction of the product graph,
$$
r_k\equiv r_{k-1}+g_i(e_k)\pmod {p_j},\qquad k=1,\ldots,\ell.
$$
Therefore,
$$
r_\ell\equiv\sum_{k=1}^{\ell}g_i(e_k)=g_i(Q)\not\equiv0\pmod {p_j},
$$
where $Q=e_1\cdots e_\ell$ is the walk in $R$.
For the three families of shortest paths, $Q$ is a residual closed walk, an $s$-$t$ walk, or a $t$-$s$ walk respectively.

If $Q$ is closed, it can be decomposed into cycles $C_1,\ldots,C_m$ satisfying
$$
g_i(Q)=\sum_{k=1}^{m}g_i(C_k),\qquad \bar w(Q)=\sum_{k=1}^{m}\bar w(C_k).
$$
Otherwise, it can be decomposed into one path $P$ and cycles $C_1,\ldots,C_m$ satisfying
$$
g_i(Q)=g_i(P)+\sum_{k=1}^{m}g_i(C_k),\qquad \bar w(Q)=\bar w(P)+\sum_{k=1}^{m}\bar w(C_k).
$$
Since every term on the right-hand side of each cost identity is nonnegative, the decomposition does not increase the cost of any component.
Moreover, $g_i(Q)\not\equiv0\pmod {p_j}$ implies that there exists a component $C\in\{C_1,\ldots,C_m\}$ or $C=P$ such that $g_i(C)\not\equiv0\pmod {p_j}$.
Thus at least one component $C$ satisfies $\chi^C\notin L$.
Since $\bar w$ is nonnegative, we have $\bar w(C)\leq \bar w(Q)$.
Moreover, every component can be used to augment $f^0$, so its original residual cost is nonnegative by the optimality of $f^0$.
The additivity of the original residual costs therefore also gives $w'(C)\leq w'(Q)$.
Let $\mathcal C$ be the collection of all such $C$ obtained from all the computed shortest walks.

We now show that we can either obtain a minimum cost $W \in \mathcal{W}(f^0)$ satisfying $\chi^W\notin L$ from $\mathcal{C}$, or certify that no such $W$ exists.
We first suppose that $\mathcal C\neq\emptyset$.
For each $C\in\mathcal C$, let $x_C$ and $y_C$ be its endpoints.
We recover its original residual cost by \eqref{eq:re-weight}, and choose $W^*\in\mathop{\arg\min}_{C\in\mathcal C}w'(C)$.
We next show that $W^*$ is a minimum cost member of $\mathcal W(f^0)$ with $\chi^{W^*}\notin L$.
Let $W\in\mathcal W(f^0)$ satisfy $\chi^W\notin L$.
Every nonzero integer of absolute value at most $Z$ is nonzero modulo at least one of $p_1,\ldots,p_K$.
Indeed, if such an integer were divisible by all these primes, their product would divide it, whereas
$$
\prod_{j=1}^Kp_j\ge2^K>Z.
$$
Thus, every $W\in\mathcal W(f^0)$ with $\chi^W\notin L$ has some coordinate $g_i(W)\neq0$ with absolute value at most $Z$, and it corresponds to a walk in the corresponding product graph.
Choose $j$ such that $g_i(W)\not\equiv0\pmod {p_j}$.
Suppose that $W=e_1,\ldots,e_\ell$ where $e_k=(v_{k-1},v_k)$, $k=1,\ldots,\ell$.
We can construct a path $W'$ in the product graph as $W'=e_1',\ldots,e_\ell'$ where $r_0=0$, $e_k'=(\pstate{v_{k-1}}{r_{k-1}},\pstate{v_k}{r_k})$ and $r_k=\sum_{b=1}^k g_i(e_b)\pmod{p_j}$, $k=1,\ldots,\ell$.
Let $C'$ be a shortest path in the product graph from $\pstate{v_0}{0}$ to $\pstate{v_\ell}{r_\ell}$.
We have $\bar w(C')\leq \bar w(W')$.
Let $Q$ be the residual walk corresponding to $C'$, and let $C\in\mathcal C$ be a component of its decomposition with $\chi^C\notin L$.
Since $Q$ and $W$ have the same endpoints, \eqref{eq:re-weight} gives $w'(Q)\leq w'(W)$.
The preceding decomposition argument gives $w'(C)\leq w'(Q)\leq w'(W)$.
Therefore, we have $w'(W^*)\leq w'(C)\leq w'(W)$.
Since $W$ was arbitrary, $W^*$ is the minimum cost path or cycle in $\mathcal{W}(f^0)$ with $\chi^{W^*}\notin L$.

Conversely, we suppose that $\mathcal C=\emptyset$.
If some $W\in\mathcal W(f^0)$ satisfied $\chi^W\notin L$, we can construct a path in the product graph whose end point has a nonzero second coordinate using similar method, contradicting $\mathcal{C}=\emptyset$.
Thus no $W\in\mathcal W(f^0)$ satisfies $\chi^W\notin L$.

Finally, the function $d$, the product graphs, and the length of the shortest path can be computed in polynomial time.
There are $hK$ product graphs, each with $|V|p_j$ vertices and $|E|p_j$ arcs, and the first $K$ primes are polynomially bounded in $K$ and can be generated in polynomial time.
Therefore, the above procedure can be completed in polynomial time.
\end{proof}

We can now solve the minimum cost $L$-nontrivial flow problem.

\begin{theorem}\label{thm:L-flow-poly}
The minimum cost $L$-nontrivial flow problem can be solved in polynomial time.
\end{theorem}

\begin{proof}
Let $f^0$ be a minimum cost flow of $D$ in $\mathcal F$.
If $f^0\notin L$, then $f^0$ is a minimum cost $L$-nontrivial flow because it is optimal among all flows in $\mathcal F$.
Suppose that $f^0\in L$.
Every flow $f\in\mathcal F$ can be obtained from $f^0$ by augmenting along a collection of residual directed cycles and directed $s$-$t$ or $t$-$s$ paths.
Thus, there exist some $W_1,\ldots,W_k\in\mathcal W(f^0)$ such that
\begin{equation}\label{eq:f-f^0}
f-f^0=\sum_{j=1}^k \chi^{W_j},\qquad w(f)-w(f^0)=\sum_{j=1}^kw'(W_j).
\end{equation}
Since $f^0$ is a minimum cost flow over $\mathcal F$, we have $w'(W_j)\geq 0$, $j=1,\ldots,k$.

By Lemma~\ref{lem:residual-L-flow}, we can find a minimum cost $W\in \mathcal{W}(f^0)$ with $\chi^W\notin L$, or certify that no such $W$ exists in polynomial time.
By~\eqref{eq:f-f^0}, $f$ can be obtained from $f^0$ by a sequence of augmentations along $W_1,\ldots,W_k$.
If no such $W$ exists, then every $W\in\mathcal{W}(f^0)$ satisfies $\chi^W\in L$.
Thus, every $f\in\mathcal{F}$ is $L$-trivial by \eqref{eq:f-f^0} and the linearity of $L$.

Conversely, suppose that there exists a minimum cost $W\in\mathcal{W}(f^0)$ with $\chi^W\notin L$.
Let $f^*$ be the flow obtained from $f^0$ by augmenting along $W$.
By Lemma~\ref{lem:augment-L-nontrivial}, $f^*$ is $L$-nontrivial.
For any $L$-nontrivial flow $f$, \eqref{eq:f-f^0} implies that $\chi^{W_j}\notin L$ for some $j$, and hence $w'(W)\leq w'(W_j)$.
Since all $w'(W_j)$ are nonnegative, we obtain
$$
w(f^*)=w(f^0)+w'(W)
=w(f)+w'(W)-\sum_{j=1}^k w'(W_j)
\leq w(f).
$$
Therefore, $f^*$ is a minimum cost $L$-nontrivial flow.
\end{proof}

\subsection{Proof of Theorem~\ref{thm:nucleolus-general}}\label{subsec:delta}

In this section, we prove Theorem~\ref{thm:nucleolus-general}.
By the preceding analysis, it suffices to show that $A(x,L)$ and $B(x,L)$ can be computed in polynomial time.
We first show that $A(x,L)$ can be computed in polynomial time.

\begin{lemma}\label{lem:A-poly}
Given $x\in\mathbb Q^N$ and a basis in $\mathbb Q^N$ of a subspace $L\subseteq\mathbb R^N$, $A(x,L)$ and a minimizing flow can be computed in polynomial time.
\end{lemma}

\begin{proof}
Define $w(e)=x(e)-c(e)$ for every $e\in N$ and $w(e)=-c(e)$ for every $e\in N^0$.
Let $L'=\{z\in\mathbb R^E\mid z_N\in L\}$.
It is easy to see that $L'$ is a linear subspace of $\mathbb{R}^E$.
A rational basis of $L'$ is obtained by extending a basis of $L$ by zero on $N^0$ and adjoining the coordinate unit vectors of all public arcs.
Then $f_N\notin L$ if and only if $f\notin L'$, so $A(x,L)=\min\{w(f)\mid f\in\mathcal F,~f\notin L'\}$.
The claim follows from Theorem~\ref{thm:L-flow-poly}.
\end{proof}

It remains to show that we can compute $B(x,L)$ in polynomial time.

\begin{lemma}\label{lem:B-poly}
Given $x\in\mathbb{R}_+^N$ and a basis of a subspace $L\subseteq\mathbb R^N$, $B(x,L)$ and a minimizing pair $(e,f)$ can be computed in polynomial time where $e\in N$ and $f\in\mathcal{F}$.
\end{lemma}

\begin{proof}
Given $e\in N$ with $\chi_e\notin L$, we compute the flow $f\in\mathcal{F}$ with $f(e)=0$ that minimizes $B(x,L)$.
Let $D'$ be obtained from $D$ by deleting arc $e$.
We then obtain $f$ by solving the minimum cost flow problem on $D'$ with cost $x-c$.
Selecting the minimum of $x(e)+x(f)-c(f)$ over at most $|N|$ choices shows that $B(x,L)$ and a minimizing pair can be computed in polynomial time.
\end{proof}

We now complete the proof of Theorem~\ref{thm:nucleolus-general}.

\begin{proof}[Proof of Theorem~\ref{thm:nucleolus-general}]
We first handle $v(N)=0$.
If $v(N)=0$, then $v(S)=0$ for every $S\subseteq N$.
Thus, the zero allocation is the nucleolus.
The polynomial-size formulation in Theorem~\ref{thm:ls-core-des-v>1} gives such an oracle.
Because all arc capacities are one, $v(S)$ is an integer satisfying $0\leq v(S)\leq |E|$ for every $S\subseteq N$ and therefore has encoding length $O(\log |E|)$.
In both cases, $P_1(\epsilon_1)\subseteq\mathbb R_+^N$ by Lemma~\ref{lem:ls-core-non-neg}.
According to Lemma~\ref{lem:Delta-decomposition}, Lemma~\ref{lem:A-poly} and Lemma~\ref{lem:B-poly}, we can compute $\Delta(x,L)$ and a minimizing coalition for every $x\in P_1(\epsilon_1)\cap\mathbb Q^N_+$ and every rational basis of a subspace $L$ containing $\chi_N$.
Thus the nucleolus can be computed in polynomial time by Theorem~\ref{thm:Delta-recursion}.
\end{proof}

\section{Conclusion} \label{sec:conclusion}

We have given an efficient description of the least core and a polynomial-time algorithm for computing the nucleolus of simple flow games.
The algorithm applies regardless of whether the core is empty and thus generalizes the result of Potters et al.~\cite{PR06}.

Simple flow games naturally generalize to flow games with general capacities, for which it is well known that computing the nucleolus is NP-hard~\cite{DF09}.
When core is non-empty, Potters et al.~\cite{PR06} give a combinatorial algorithm for the nucleolus of simple flow games.
Thus an interesting direction is to develop a combinatorial algorithm for computing the nucleolus of simple flow games when core is empty.
On the other hand, we note that the polynomial-time algorithm for the minimum cost $L$-nontrivial flow problem defined in this paper remains valid for networks with general capacities.
It would also be interesting to investigate sufficient conditions under which the nucleolus of a flow game with general capacities can be computed in polynomial time.

Gangam et al.~\cite{GG24} have shown that the leximin dual-consistent core of general-capacity flow games without public arcs, a relaxation of the nucleolus, can be computed in polynomial time.
Despite this, the complexity of computing the leximin dual-consistent core of general-capacity flow games with public arcs remains open.

\bibliographystyle{splncs04}
\bibliography{ref.bib}
\nocite{*}

\newpage
\appendix

\section{The case where $v(N)=1$} \label{app:ls-core-nue-v=1}

We assume $v(N)=1$ throughout this appendix section.
When $\mathcal{C}(N,v)\neq\emptyset$, the nucleolus is computable by the algorithm of Potters et al.~\cite{PR06}.
Hence we focus on $\mathcal{C}(N,v)=\emptyset$.
For a game $(N,g)$, write $e_g(x,S)=x(S)-g(S)$ and denote its nucleolus by $\eta(N,g)$.
Thus $e_v(x,S)=e(x,S)$ in the notation of Subsection~\ref{sec:coop-game}.

We first characterize the least core by analyzing cuts and introducing an auxiliary game.
Let $\mathcal P$ be the set of paths in $D$, and write $P_N=P\cap N$ for $P\in\mathcal P$.
Define $k=\min\{|C|\mid C\subseteq N,\ C\text{ is a cut}\}$ and $\mathcal C^*=\{C\subseteq N\mid C\text{ is a cut and }|C|=k\}$.
The parameter $k$ satisfies
\begin{equation}\label{eq:cut-bounds}
2\le k\le |N|.
\end{equation}
The upper bound holds because $N$ is a cut.
For the lower bound, suppose that $\{e\}\subseteq N$ is a cut.
Every coalition of value one then contains $e$, so $\chi_e\in\mathcal{C}(N,v)$, contrary to $\mathcal{C}(N,v)=\emptyset$.
The max-flow min-cut theorem and $v(N)=1$ imply that every minimum cut in $D$ consists of one arc.
Equation~\eqref{eq:cut-bounds} excludes a cut consisting of one private arc, so every minimum cut in $D$ consists of one public arc.

We construct the auxiliary game as follows.
We replace every public arc $e$ with $k$ parallel public arcs $e_1,\ldots,e_k$, each of capacity one, and retain every private arc with capacity one.
The resulting network is $D_k=(V,N\cup N^0_k)$ with source $s$ and sink $t$, where $N^0_k=\{e_j\mid e\in N^0,\ 1\le j\le k\}$.
The auxiliary game $(N,\overline v)$ has the same player set $N$.
For each coalition $S\subseteq N$, $\overline v(S)$ is defined by the value of a maximum flow from $s$ to $t$ in the subgraph of $D_k$ induced by $N^0_k$ and $S$.
Lemma~\ref{lem:aux-core} identifies the core of the auxiliary game through the set $\mathcal C^*$ of minimum cuts contained in $N$.

\begin{lemma}\label{lem:aux-core}
$\overline v(N)=k$, $\mathcal{C}(N,\overline v)=\operatorname{conv}\{\chi_C\mid C\in\mathcal C^*\}$.
\end{lemma}

\begin{proof}
Consider an $s$-$t$ cut $T$ of $D_k$.
If $T$ contains a public arc $e\in N^0_k$, then it must contain all $k$ parallel copies of this arc and has size at least $k$.
If $T$ contains no public arc, then it has size at least $k$ by the definition of $k$.
Therefore, $|T|\geq k$ and the size of the minimum cut of $D_k$ is $k$.
So the maximum flow minimum cut theorem gives $\overline v(N)=k$.
By the definition of $k$ and $\mathcal{C}^*$, there is a minimum cut without public arcs in $D_k$.
Thus the core of $(N,\overline v)$ is non-empty and $\mathcal{C}(N,\overline v)$ is the convex hull of the characteristic vectors of its minimum cuts without public arcs by using Theorem~\ref{thm:core-cut}.
\end{proof}

We next characterize the least core of the original game by using the auxiliary game.

\begin{theorem}\label{thm:least-core-reduction}
$\epsilon_1=\frac1k-1$, $P_1(\epsilon_1)=\frac1k\mathcal{C}(N,\overline v)$.
\end{theorem}

\begin{proof}
Set $Q=k^{-1}\mathcal{C}(N,\overline v)$.
We first show that every $x\in Q$ is feasible for $(P_1)$ with $\epsilon=1/k-1$.
Let $y=kx\in\mathcal{C}(N,\overline v)$.
Then $y\geq0$ and $y(S)\geq\overline v(S)$ for every $S\subseteq N$.
For every proper coalition $S\subset N$, we have $v(S)=1$ if and only if $\overline v(S)\geq1$.
If $v(S)=1$, then $x(S)\geq1/k$.
If $v(S)=0$, then $x(S)\geq0$.
Thus $x(S)-v(S)\geq1/k-1$ in either case.
Since $x(N)=\overline v(N)/k=1$, it follows that $\epsilon_1\geq1/k-1$.

We next prove that $P_1(\epsilon_1)\subseteq\mathbb R_+^{|N|}$.
Suppose that $x\in P_1(\epsilon_1)$ and $x(e')<0$ for some $e'\in N$.
The constraint for $\{e'\}$ gives $\epsilon_1\leq x(e')-v(\{e'\})<0$.
Let $S\subseteq N\setminus\{e'\}$.
If $S\neq N\setminus\{e'\}$, then
$$
e_v(x,S)>x(S\cup\{e'\})-v(S)\geq e_v(x,S\cup\{e'\})\geq\epsilon_1.
$$
If $S=N\setminus\{e'\}$, then
$$
e_v(x,S)>x(N)-v(S)\geq x(N)-v(N)=0>\epsilon_1.
$$
Thus $e_v(x,S)>\epsilon_1$ for every $S\subseteq N\setminus\{e'\}$.
Set $\delta=\min\{e_v(x,S)-\epsilon_1\mid S\subseteq N\setminus\{e'\}\}>0$, and define $x'$ by
$$
x'(e)=\begin{cases}x(e)+\frac{\delta}{2},&\text{if }e=e',\\
x(e)-\frac{\delta}{2(|N|-1)},&\text{otherwise.}\end{cases}
$$
Then $x'(N)=v(N)$.
For every proper coalition $S$ containing $e'$, we have $x'(S)>x(S)$ and hence $e_v(x',S)>\epsilon_1$.
For every coalition $S$ excluding $e'$, we have $e_v(x',S)\geq e_v(x,S)-\delta/2\geq\epsilon_1+\delta/2>\epsilon_1$.
Thus $e_v(x',S)>\epsilon_1$ for every $S\subset N$, contradicting the optimality of $\epsilon_1$.
Therefore $P_1(\epsilon_1)\subseteq\mathbb R_+^{|N|}$.

Now fix $x\in P_1(\epsilon_1)$.
The lower bound $\epsilon_1\geq1/k-1$ gives $1+\epsilon_1>0$, while $\mathcal{C}(N,v)=\emptyset$ implies $\epsilon_1<0$.
If $P_N\neq N$, the constraint for $P_N$ in $(P_1)$ gives $x(P_N)\geq1+\epsilon_1$.
If $P_N=N$, the same inequality follows from $x(N)=1$.
Hence every path $P\in\mathcal P$ satisfies $x(P_N)\geq1+\epsilon_1$.

For any $S\subseteq N$, choose an integral maximum flow in the subgraph of $D_k$ with arc set $S\cup N_k^0$.
After removing cycles, decompose it into $r=\overline v(S)$ unit flows along paths.
The projections $P^1,\ldots,P^r$ of these paths belong to $\mathcal P$, and their private arc sets $P^1_N,\ldots,P^r_N$ are pairwise disjoint and contained in $S$.
Since $x\geq0$,
$$
x(S)\geq\sum_{j=1}^{r}x(P^j_N)\geq r(1+\epsilon_1)=(1+\epsilon_1)\overline v(S).
$$
Taking $S=N$ and using $x(N)=1$ and $\overline v(N)=k$ gives $1\geq k(1+\epsilon_1)$.
Consequently, $\epsilon_1\leq1/k-1$, so equality holds.

Therefore $kx(S)\geq\overline v(S)$ for every $S\subseteq N$.
As $kx(N)=k=\overline v(N)$, we have $kx\in\mathcal{C}(N,\overline v)$, so $x\in Q$.
The feasibility established in the first paragraph now gives $Q\subseteq P_1(\epsilon_1)$ as well.
Hence $P_1(\epsilon_1)=Q$.
\end{proof}

We next compute the nucleolus in polynomial time by relating it to the nucleolus of the auxiliary game.

\begin{theorem}\label{thm:nucleolus-scaling}
$\eta(N,v)=\frac1k\eta(N,\overline v)$.
\end{theorem}

\begin{proof}
Let $(\overline P_r)$ denote the counterpart of $(P_r)$, with optimal value $\overline\epsilon_r$ and optimal allocation set $\overline P_r(\overline\epsilon_r)$ for the auxiliary game.
Lemma~\ref{lem:aux-core} gives a non-empty core of the auxiliary game, while the empty coalition constraint gives $\overline\epsilon_1\le0$.
Hence $\overline\epsilon_1=0$ and $\overline P_1(\overline\epsilon_1)=\mathcal{C}(N,\overline v)$.
Theorem~\ref{thm:least-core-reduction} gives $\epsilon_1=1/k-1$ and $P_1(\epsilon_1)=\mathcal{C}(N,\overline v)/k$.
Thus $P_1(\epsilon_1)=\overline P_1(\overline\epsilon_1)/k$ and $\overline\epsilon_1=k\epsilon_1+k-1$.
For $r\ge2$, consider LP~\eqref{lp:pr} and its auxiliary counterpart $(\overline P_r)$.
Fix $Q\subseteq \mathcal{C}(N,\overline v)$ containing at least two points.
If their preceding optimal allocation sets are $Q/k$ and $Q$, respectively, then both linear programs use the coalitions $S\notin\text{Fix}(Q)$ because $\text{Fix}(Q)=\text{Fix}(Q/k)$.
Let $\delta_Q(y)$ be the minimum of $y(P_N)-1$ over $P\in\mathcal P$ with $P_N\notin\text{Fix}(Q)$ and $y(e)$ over $e\in N$ with $\{e\}\notin\text{Fix}(Q)$.
These functions are nonnegative on $\mathcal{C}(N,\overline v)$, and $\{e\}\notin\text{Fix}(Q)$ for some $e\in N$.
We will prove
\begin{equation}\label{eq:stage-objectives}
\min_{S\notin\text{Fix}(Q)}e_{\overline v}(y,S)=\delta_Q(y)=k\min_{S\notin\text{Fix}(Q)}e_v(y/k,S)+k-1, \forall y\in Q.
\end{equation}

If \eqref{eq:stage-objectives} holds, the theorem follows by induction on $r$.
Suppose for some $r\ge2$ that $P_{r-1}(\epsilon_{r-1})=\overline P_{r-1}(\overline\epsilon_{r-1})/k$.
Equation~\eqref{eq:stage-objectives} says that the largest feasible objective value in $(\overline P_r)$ at $y\in \overline P_{r-1}(\overline\epsilon_{r-1})$ is $k$ times that in LP~\eqref{lp:pr} at $y/k$ plus $k-1$.
Since $k>0$, their optimal values and optimal allocation sets satisfy
$$
\overline\epsilon_r=k\epsilon_r+k-1,\qquad P_r(\epsilon_r)=\frac1k\overline P_r(\overline\epsilon_r).
$$
Both identities hold for $r=1$, so induction proves them for every $r$ until termination.
The terminal optimal sets are singletons containing $\eta(N,v)$ and $\eta(N,\overline v)$, respectively, and hence $\eta(N,v)=\eta(N,\overline v)/k$.

It remains to prove \eqref{eq:stage-objectives}.
For $S\subseteq N$, decompose an integral maximum flow of $D_k[S]$, after removing cycles, into $m=\overline v(S)$ paths.
Their projections $P^1,\ldots,P^m$ have pairwise disjoint private arc sets $P^j_N=P^j\cap N\subseteq S$, so
\begin{equation}\label{eq:aux-excess-decomposition}
e_{\overline v}(y,S)=\sum_{j=1}^{m}\bigl(y(P^j_N)-1\bigr)+\sum_{e\in S\setminus\bigcup_{j=1}^{m}P^j_N}y(e).
\end{equation}
Since $v(S)=1$ implies $\overline v(S)\ge1$, we have $m\ge1$.
Using $P^1$ in this case gives
\begin{equation}\label{eq:original-excess-decomposition}
k e_v(y/k,S)+k-1=
\begin{cases}
y(P^1_N)-1+\displaystyle\sum_{e\in S\setminus P^1_N}y(e),&v(S)=1,\\[4pt]
k-1+\displaystyle\sum_{e\in S}y(e),&v(S)=0.
\end{cases}
\end{equation}
Empty sums and unions in \eqref{eq:aux-excess-decomposition} have their usual meanings.
If $S\notin\text{Fix}(Q)$, each decomposition in \eqref{eq:aux-excess-decomposition} and \eqref{eq:original-excess-decomposition} contains a summand $y(P^j_N)-1$ with $P^j_N\notin\text{Fix}(Q)$ or a summand $y(e)$ with $\{e\}\notin\text{Fix}(Q)$.
That summand is at least $\delta_Q(y)$, and the remaining terms are nonnegative.
Hence, for every $y\in Q$, $\min_{S\notin\text{Fix}(Q)}e_{\overline v}(y,S)\ge\delta_Q(y)$ and $k\min_{S\notin\text{Fix}(Q)}e_v(y/k,S)+k-1\ge\delta_Q(y)$.

The reverse inequalities follow by comparing each function defining $\delta_Q$ with suitable coalition excesses.
For a path $P$ with $P_N\notin\text{Fix}(Q)$, the inequalities $\overline v(P_N)\ge1$ and $v(P_N)=1$ give
\begin{equation}\label{eq:path-comparison}
e_{\overline v}(y,P_N)\le y(P_N)-1=k e_v(y/k,P_N)+k-1.
\end{equation}
For $e\in N$ with $\{e\}\notin\text{Fix}(Q)$, fix a maximum flow of $D_k$ and let $T_1,\ldots,T_k$ be the disjoint private arc sets of its paths.
For every $y\in \mathcal{C}(N,\overline v)$, we have $y(T_j)\ge1$ and $\sum_j y(T_j)\le y(N)=k$, so $y(T_j)=1$ for all $j$.
Since $k\ge2$, some $T_j$ avoids $e$.
Using $\overline v(\{e\})\ge0$ and $v(T_j\cup\{e\})=1$, we obtain
\begin{equation}\label{eq:coordinate-representatives}
e_{\overline v}(y,\{e\})\le y(e)=k e_v(y/k,T_j\cup\{e\})+k-1.
\end{equation}
Since $T_j\in\text{Fix}(Q)$ and $\{e\}\notin\text{Fix}(Q)$, both $\{e\}$ and $T_j\cup\{e\}$ lie outside $\text{Fix}(Q)$.
Equations~\eqref{eq:path-comparison} and \eqref{eq:coordinate-representatives} show that, for every $y\in Q$, $\min_{S\notin\text{Fix}(Q)}e_{\overline v}(y,S)\le\delta_Q(y)$ and\linebreak[4]
$k\min_{S\notin\text{Fix}(Q)}e_v(y/k,S)+k-1\le\delta_Q(y)$.
Together with the preceding inequalities, this proves \eqref{eq:stage-objectives}.
\end{proof}

Theorem~\ref{thm:nucleolus-scaling} yields the following computational consequence.

\begin{corollary}\label{cor:poly-nucleolus}
The nucleolus $\eta(N,v)$ can be computed in polynomial time.
\end{corollary}

\begin{proof}
We first show that $k$ can be computed in polynomial time.
We assign capacity one to each private arc and capacity $|N|+1$ to each public arc in $D$, and compute a minimum cut $T$ of the new network.
Since $N$ is a cut of $D$, we have $|T|\leq |N|$.
Thus $T$ contains no public arcs, and $|T|=k$ by the definition of $k$.
Therefore, the auxiliary network $D_k$ can be constructed in polynomial time.
Lemma~\ref{lem:aux-core} ensures that the auxiliary game has a non-empty core.
Thus $\eta(N,\overline v)$ can be computed by the polynomial-time algorithm of Potters et al.~\cite{PR06}.
Moreover, we obtain $\eta(N,v)$ by Theorem~\ref{thm:nucleolus-scaling}.
\end{proof}

\section{Omitted proofs}\label{app:prove}
\subsection{Proof of Lemma \ref{lem:ls-core-non-neg}} \label{app:ls-core-non-neg}

\begin{proof}
We first note that $\epsilon_1\leq 0$, by the constraint for the empty coalition in ($P_1$).
We next prove that $P_1(\epsilon_1)\subseteq \mathbb{R}^{|N|}_+$.
Suppose that $x\in P_1(\epsilon_1)$ and that there exists $e_0\in N$ such that $x(e_0)<0$.
In this case, $\epsilon_1\leq x(e_0)-v(\{e_0\})<0$.
Let $S\subseteq N\setminus\{e_0\}$.
If $S\neq N\setminus \{e_0\}$, then
$$
e(x,S)>x(S\cup \{e_0\})-v(S)\geq e(x,S\cup\{e_0\})\geq \epsilon_1.
$$
If $S=N\setminus \{e_0\}$, then
$$
e(x,S)>x(N)-v(S)\geq x(N)-v(N)=0>\epsilon_1.
$$
In summary, $e(x,S)>\epsilon_1$ for any $S\subseteq N\setminus\{e_0\}$.
Let $x'\in \mathbb{R}^{|N|}$ be a vector satisfying
$$
x'(e)=\begin{cases}
x(e)+\frac{\delta}{2},~\text{if }e=e_0, \\
x(e)-\frac{\delta}{2(|N|-1)},~\text{otherwise.}
\end{cases}
$$
where $\delta=\min\{e(x,S)-\epsilon_1\mid S\subseteq N\setminus\{e_0\}\}$.
Therefore, we have $x'(N)=v(N)$, $x'(S)>v(S)+\epsilon_1$, $\forall S\subset N$, i.e., $x'\in P_1(\epsilon_1)$.
However, $\min\{e(x',S)\mid S\subset N\}>\epsilon_1$, which contradicts the fact that $\epsilon_1$ is the optimal objective function value of the linear program ($P_1$).

\end{proof}

\subsection{Proof of Lemma \ref{lem:ls-lp-replace}} \label{app:ls-lp-replace}

\begin{proof}
Define ($P_1'$) as the linear program obtained by replacing the inequalities in the linear program ($P_1$), and let $\epsilon_1'$ be the optimal objective function value of ($P'_1$).
We prove that $\epsilon_1=\epsilon_1'$, and $P_1(\epsilon_1)=P'_1(\epsilon_1')$.
Let $x\in P_1(\epsilon_1)$.
By Lemma \ref{lem:ls-core-non-neg}, $x(e)\geq 0$, $\forall e\in N$.
Let $f\in \mathcal{F}$.
Let $S_f=\{e\in N\mid f(e)=1\}$.
If $S_f\subsetneq N$, then $x(f)-c(f)\ge x(S_f)-v(S_f)\ge \epsilon_1$.
If $S_f=N$, then $x(f)=x(N)=v(N)$, and since $c(f)\le v(N)$, we have $x(f)-c(f)\ge 0\ge \epsilon_1$.
Therefore, $x(f)-c(f)\geq \epsilon_1$ for all $f\in \mathcal{F}$, i.e., $(x,\epsilon_1)$ is a feasible solution of ($P_1'$).
This shows that $P_1(\epsilon_1)\subseteq P_1'(\epsilon_1)$ and $\epsilon_1\leq \epsilon'_1$.

On the other hand, let $x\in P'_1(\epsilon_1')$.
For any $S\subset N$, let $f\in \mathcal{F}$ be a maximum flow of $D[S\cup N^0]$.
Then $x(S)-v(S)\geq x(f)-c(f)\geq \epsilon'_1$, i.e., $(x,\epsilon_1')$ is a feasible solution of ($P_1$).
This shows that $P_1'(\epsilon_1')\subseteq P_1(\epsilon_1')$ and $\epsilon_1'\leq \epsilon_1$.

Therefore, we have shown that $\epsilon_1=\epsilon'_1$.
Since $P_1(\epsilon_1)\subseteq P_1'(\epsilon_1)$ and $P_1'(\epsilon_1')\subseteq P_1(\epsilon_1')$, we obtain $P_1(\epsilon_1)=P'_1(\epsilon_1')$.
\end{proof}

\subsection{Proof of Lemma \ref{lem:opt-flow}}\label{app:opt-flow}

\begin{proof}
Let $(\pi,y)$ be an optimal solution of the dual program \eqref{lp:mincost-dual}.
For each $e\in E$, define
$$
\sigma(e)=
\begin{cases}
x(e)-\bigl(\pi(\partial^-e)-\pi(\partial^+e)+y(e)\bigr), & e\in N,\\[2mm]
-\bigl(\pi(\partial^-e)-\pi(\partial^+e)+y(e)\bigr), & e\in N^0.
\end{cases}
$$
By dual feasibility, $\sigma(e)\ge 0$ and $y(e)\le 0$, $\forall e\in E$.
Since $(\pi,y)$ is dual optimal, $g$ is optimal for \eqref{lp:mincost-flow} if and only if $g$ satisfies the complementary slackness conditions, i.e.,
$$
\sigma(e)g(e)=0,~y(e)(1-g(e))=0,~\forall e\in E.
$$

We first prove the inclusion.
If $f\in\mathcal{F}^x$, then, by the definitions of $M_x$ and $R_x$, $f(e)=1$, $\forall e\in M_x$ and $f(e)=0$, $\forall e\in R_x$.
Hence
$$
\mathcal{F}^x
\subseteq
\{f\in\mathcal{F}\mid f(e)=1,\ \forall e\in M_x,\ f(e)=0,\ \forall e\in R_x\}.
$$

It remains to prove the reverse inclusion.
Take any $f\in\mathcal{F}$ such that $f(e)=1$, $\forall e\in M_x$ and $f(e)=0$, $\forall e\in R_x$.
If $\sigma(e)>0$, then every optimal flow $g$ of \eqref{lp:mincost-flow} satisfies $g(e)=0$.
In particular, every $g\in\mathcal{F}^x$ satisfies $g(e)=0$, and hence $e\in R_x$.
Therefore,
$$
\sigma(e)f(e)=0,~\forall e\in E.
$$
Similarly, if $y(e)<0$, then every optimal flow $g$ of \eqref{lp:mincost-flow} satisfies $g(e)=1$.
In particular, every $g\in\mathcal{F}^x$ satisfies $g(e)=1$, and hence $e\in M_x$.
Therefore,
$$
y(e)(1-f(e))=0,~\forall e\in E.
$$
Thus $f$ satisfies the complementary slackness conditions with the dual optimal solution $(\pi,y)$.
Hence $f$ is an optimal solution of \eqref{lp:mincost-flow}.
Since $f\in\mathcal{F}$, we have $f\in\mathcal{F}^x$.
Consequently,
$$
\{f\in\mathcal{F}\mid f(e)=1,\ \forall e\in M_x,\ f(e)=0,\ \forall e\in R_x\}\subseteq \mathcal{F}^x.
$$
Combining the two directions, we complete the proof of the lemma.
\end{proof}
\subsection{Proof of Lemma \ref{lem:MAR-x}}\label{app:MAR-x}

\begin{proof}
Extend $x$ to $E$ by setting $x(e)=0$ for every $e\in N^0$, and let $w(e)=x(e)-c(e)$, $e\in E$.
Define $\mu(x)=\min\{w^\top f\mid f\in\mathcal F\}$.
For every $e\in E$ and $b\in\{0,1\}$, let
$$
\mu_e^b(x)=\min\{w^\top f\mid f\in\mathcal F,\ f(e)=b\},
$$
where the minimum is defined to be $+\infty$ if the corresponding feasible set is empty.

We first show that these quantities can be computed in polynomial time.
Note that minimum cost flow with integral lower and upper bounds is solvable in polynomial time (see, e.g., \cite{AM93}) and has an integral optimal solution.
Thus, $\mu(x)$ can be computed by a minimum cost flow problem with lower bounds $f(e)=0$ and upper bounds $f(e)=1$.
Moreover, $\mu_e^b(x)$ can be computed by imposing the additional lower and upper bounds $f(e)=b$.

It remains to characterize the three sets.
Since every flow $f\in\mathcal F$ satisfies either $f(e)=0$ or $f(e)=1$, we have $\mu(x)=\min\{\mu_e^0(x),\mu_e^1(x)\}$ for every $e\in E$.
By definition, $e\in M_x$ if and only if every flow in $\mathcal F^x$ uses $e$, which is equivalent to the non-existence of an optimal flow satisfying $f(e)=0$.
Hence,
$$
e\in M_x
\quad\Longleftrightarrow\quad
\mu_e^0(x)>\mu(x).
$$
Similarly,
$$
e\in R_x
\quad\Longleftrightarrow\quad
\mu_e^1(x)>\mu(x).
$$
Finally, $e\in A_x$ if and only if there are two optimal flows, one with $f(e)=0$ and one with $f(e)=1$. Therefore,
$$
e\in A_x
\quad\Longleftrightarrow\quad
\mu_e^0(x)=\mu_e^1(x)=\mu(x).
$$
Thus, $M_x$, $A_x$, and $R_x$ can be identified by solving $2|E|+1$ minimum cost flow problems, and hence in polynomial time.
\end{proof}

\subsection{Proof of Lemma \ref{lem:inter-point}} \label{app:inter-point} 

Given $x\in \mathbb{R}^{|N|}$, the following lemma states that the post-optimality problem on $\mathcal{F}^x$ can be solved in polynomial time.

\begin{small}
  \begin{algorithm}[H]
  \caption{Finding a relative interior point of the least core\label{alg:inter-point}}
  \label{alg:ALM}
  \begin{algorithmic}[1]
  \STATE {\bfseries Input:} A simple network $D=(V,E)$, $N\subseteq E$.
  \STATE {\bfseries Output:} A relative interior point $x$ in $P_1(\epsilon_1)$.
  \STATE Initialize $F=\emptyset$.
  \STATE Compute $x\in P_1(\epsilon_1)$.
  \FOR {$e\in N$}
    \STATE $y=\text{argmax}\{y(e)\mid y\in P_1(\epsilon_1)\}$.
    \STATE $x\leftarrow\frac{1}{2}(x+y)$.\label{e-step:(x+y)/2}
  \ENDFOR
  \STATE Compute $f\in\mathcal{F}^x$.
  \LOOP
    \STATE $y=\text{argmax}\{y(f)\mid y\in P_1(\epsilon_1)\}$. \label{step:max-f}
    \IF {$y(f)>x(f)$}  \label{step:y>x}
      \STATE $x\leftarrow\frac{1}{2}(x+y)$. \label{f-step:(x+y)/2}
      \STATE Compute $f'\in\mathcal{F}^x$ and $f\leftarrow f'$.
    \ELSE
      \STATE $F\leftarrow F\cup \{f\}$. \label{step:span-F}
      \STATE Write $\text{span}(F)$ as $\{d\mid Ad=0\}$, and let $a_1,\ldots,a_k$ be the rows of $A$.
      \FOR {$i\leftarrow 1$ \TO $k$}
        \STATE $f^i_1=\text{argmin}\{a_i(f')\mid f'\in \mathcal{F}^x\}$, $f^i_2=\text{argmin}\{-a_i(f')\mid f'\in \mathcal{F}^x\}$. \label{step:post-opt}
        \IF {$a_i(f^i_1)\neq 0$}
          \STATE $f\leftarrow f^i_1$, {\bfseries break}.
        \ELSIF {$a_i(f^i_2)\neq 0$}
          \STATE $f\leftarrow f^i_2$, {\bfseries break}.
        \ENDIF
      \ENDFOR
      \IF {$a_i(f^i_1)=-a_i(f^i_2)=0$ for $i=1,\ldots,k$} \label{step:aff-condi}
        \STATE {\bfseries return} $x$.
      \ENDIF
    \ENDIF
  \ENDLOOP
\end{algorithmic}
\end{algorithm}
\end{small}

\begin{lemma} \label{lem:post-opt-flow}
Given weights $w$ on the arcs, there exists a polynomial-time algorithm for finding $f \in\mathcal{F}^x$ that minimizes $w(f)$.
\end{lemma}

\begin{proof}
We first extend $x$ to $E$ by setting $x(e)=0$ for every $e\in N^0$, and let $q(e)=x(e)-c(e)$ for every $e\in E$.
We then compute $\mu=\min\{q^\top f\mid f\in\mathcal F\}$, which is a minimum cost flow problem and can be solved in polynomial time.
For any $f\in\mathcal F$, we have $f\in\mathcal F^x$ if and only if $q^\top f=\mu$.
We next encode the two-level objective as a minimum cost flow problem.
By clearing denominators, we may assume without loss of generality that $q$ and $w$ are integral.
Let $K=1+\sum_{e\in E}|w(e)|$.
Run a minimum cost flow algorithm on $D$ with costs $\bar w(e)=Kq(e)+w(e)$, $\forall e\in E$.
This is a minimum cost flow problem with integral lower and upper bounds, so it can be solved in polynomial time and admits an integral optimal solution (see, e.g., \cite{AM93}).
Let $f$ be an integer optimal flow for these costs.
If $q^\top f>\mu$, then, since $q$ is integral, $q^\top f\ge \mu+1$.
Choose any $g\in\mathcal F^x$.
Then $q^\top g=\mu$, and hence
$$
\bar w(f)-\bar w(g)
=K(q^\top f-q^\top g)+(w^\top f-w^\top g)
\ge K-\sum_{e\in E}|w(e)|>0,
$$
which contradicts the optimality of $f$.
Therefore $q^\top f=\mu$, and hence $f\in\mathcal F^x$.
If there were $g\in\mathcal F^x$ with $w^\top g<w^\top f$, then $\bar w(g)<\bar w(f)$ because $q^\top g=q^\top f=\mu$, again contradicting the optimality of $f$.
Thus $f$ minimizes $w(f)$ over all $f\in\mathcal F^x$.
\end{proof}

\begin{proof}[Proof of Lemma \ref{lem:inter-point}]
The algorithm is presented in Algorithm \ref{alg:ALM}.
We first show that Algorithm \ref{alg:ALM} constructs a relative interior point of the least core, before bounding its complexity.

\setcounter{claim}{0}
\begin{claim}
If the algorithm returns $x$, then $x$ is a relative interior point of $P_1(\epsilon_1)$.
\end{claim}

\begin{innerproof}
Let $x$ be the vector output by the algorithm.
Step~\ref{e-step:(x+y)/2} ensures that every component of $x$ is as positive as possible.
Moreover, if $x(e)>0$, then $\frac{1}{2}(x+y)(e)>0$ for any $e\in N$.
Therefore, Step~\ref{f-step:(x+y)/2} preserves the positive components of $x$.
Let $f\in \mathcal{F}^x$.
We prove that $f$ is a universal flow.
Let $y\in P_1(\epsilon_1)$.
According to Step \ref{step:aff-condi}, we have $a_i(f)=0$ for $i=1,\ldots,k$; therefore, $f\in \text{span}(F)$.
Let $F=\{f_1,\ldots,f_l\}$.
Then there exists $\alpha_1,\ldots,\alpha_l$ such that $f=\sum_{i=1}^l\alpha_if_i$.
Note that any flow in $F$ is universal.
Thus,
$$
\epsilon_1+c(f)=x(f)=\sum_{i=1}^l\alpha_ix(f_i)=\sum_{i=1}^l\alpha_iy(f_i)=y(f).
$$
This implies that $f\in\mathcal{F}^y$.
Since $y$ is an arbitrary allocation in the least core, it follows that $f$ is a universal flow.
\end{innerproof}

\begin{claim}
This algorithm can be implemented to run in polynomial time.
\end{claim}

\begin{innerproof}
Let $x'=(x+y)/2$.
Let $Q_x=\operatorname{conv}(\mathcal F^x)$.
If $y(f)>x(f)$ holds in step \ref{step:y>x}, then $\mathcal{F}^{x'}\subsetneq \mathcal{F}^x$.
Indeed, for any $h\in\mathcal F^{x'}$, we have $c(h)+\epsilon_1=x'(h)=(x(h)+y(h))/2$.
Since $x,y\in P_1(\epsilon_1)$, both $x(h)$ and $y(h)$ are at least $c(h)+\epsilon_1$.
Thus $x(h)=y(h)=c(h)+\epsilon_1$, and hence $h\in\mathcal F^x$.
On the other hand, $f\in\mathcal F^x$ and
$$
x'(f)=\frac{x(f)+y(f)}{2}>x(f)=c(f)+\epsilon_1,
$$
so $f\notin\mathcal F^{x'}$.
For $h\in\mathcal F^x$, the equality $y(h)=c(h)+\epsilon_1$ holds if and only if $h\in\mathcal F^{x'}$.
Therefore, $\operatorname{conv}(\mathcal F^{x'})$ is the face of $Q_x$ induced by the following valid inequality
$$
\sum_{e\in N}y(e)z(e)-\sum_{e\in E}c(e)z(e)\ge \epsilon_1.
$$
This face is proper because it does not contain $f$.
Therefore, $\dim \operatorname{conv}(\mathcal F^{x'})<\dim Q_x$.
Since this dimension is at most $|E|$ and strictly decreases whenever step \ref{f-step:(x+y)/2} is executed, step \ref{f-step:(x+y)/2} runs at most $|E|$ times.

On the other hand, we assume that $y(f)\leq x(f)$.
Since $a_i(f')\neq 0$ for some $i$ and $f'\in\mathcal{F}^x$, we have $f'\notin \text{span}(F)$.
Thus the dimension of $\text{span}(F)$ strictly increases when we add $f'$ to $F$.
Since the dimension of $\text{span}(F)$ is at most $|E|$, step \ref{step:span-F} runs at most $|E|$ times.

Finally, we show that each step of the algorithm can be computed in polynomial time.
According to Lemma \ref{lem:ls-lp-replace}, we can separate the polytope $P_1(\epsilon_1)$ in polynomial time.
Therefore, the optimization problem in step \ref{step:max-f} can be solved in polynomial time.
According to Lemma \ref{lem:post-opt-flow}, we can solve the post-optimality problem in step \ref{step:post-opt} in polynomial time.
Since the number of linearly independent vectors orthogonal to $\text{span}(F)$ does not exceed $|E|$, step \ref{step:post-opt} runs at most $|E|$ times.
\end{innerproof}
\end{proof}

\end{document}